\documentclass[preprint,12pt]{elsarticle}

\usepackage[utf8]{inputenc}
\usepackage{amsmath,amsfonts,amssymb,amsthm,mathtools,mathrsfs,booktabs,array,geometry,graphicx,subcaption,float,placeins}
\usepackage[ruled,vlined]{algorithm2e}
\usepackage{enumitem}
\usepackage[colorlinks=true,allcolors=blue]{hyperref}
\graphicspath{{figures/}}
\numberwithin{equation}{section}
\newtheorem{defn}{Definition}[section]

\newtheorem{prop}[defn]{Proposition}
\newtheorem{cor}[defn]{Corollary}
\newtheorem{lem}[defn]{Lemma}

\newtheorem{asu}[defn]{Assumption}
\newcommand{\D}{{\mathop{}\!\mathrm{d}}}

\newcommand{\R}{\mathbb{R}}
\newcommand{\Q}{\mathbb{Q}}

\newcommand{\K}{\mathbb{K}}
\newcommand{\N}{\mathbb{N}}

\newcommand{\E}{\mathbb{E}}

\newcommand{\dee}{\mathrm{d}}

\journal{Finance Research Letters}

\begin{document}

\begin{frontmatter}
\title{Fast catastrophe bond valuation with neural-network surrogates}

\author[aff1]{Julian Sester\texorpdfstring{\footnote{Corresponding author, email: jul\_ses@nus.edu.sg}}{}}
\author[aff1]{Huansang Xu}
\affiliation[aff1]{organization={Department of Mathematics, National University of Singapore},
addressline={21 Lower Kent Ridge Road}, postcode={119077}, country={Singapore}}

\begin{abstract}
Catastrophe bonds are increasingly important risk-transfer securities, but structural pricing is too slow for real-time valuation, screening, and sensitivity analysis. We develop a neural-network surrogate for the pricing operator of a compound-Poisson catastrophe bond model. Training labels are generated by Monte Carlo simulation with importance sampling, so the network learns variance-reduced structural prices rather than sparse market quotes which are difficult to extrapolate reliably. Across Gamma and Lognormal severity specifications, the selected networks attain very small absolute approximation error{ on the stated training domain}. After training, 1000 contracts are priced in about 0.03--0.04 seconds, compared with tens to hundreds of seconds for Monte Carlo with importance sampling and many hours for a partial integro-differential equation benchmark. The result is a fast {in-domain} structural valuation engine that also produces economically interpretable sensitivities to catastrophe intensity, attachment threshold, and interest rates.
\end{abstract}

\begin{keyword}
catastrophe bonds \sep neural networks \sep importance sampling \sep rare-event simulation \sep risk management
\JEL G12 \sep G22 \sep C63
\end{keyword}

\end{frontmatter}

\section{Introduction}
Catastrophe bonds (CAT bonds) are increasingly important risk-transfer securities, but structural pricing is too slow for real-time valuation, screening, and sensitivity analysis. A neural surrogate trained on rare-event simulation makes structural CAT bond valuation effectively instantaneous.

A CAT bond transfers extreme insurance losses to capital-market investors. If no trigger occurs, investors receive coupons and principal; if the trigger occurs, part or all of the principal is released to indemnify the sponsor. This payoff makes CAT bonds useful for insurers seeking funded protection and for investors seeking insurance-linked risk {which prior CAT-bond primers describe as largely uncorrelated with traditional financial markets} \citep{BraunKousky2021,Polacek2018}. However, its construction also makes repeated valuation difficult: prices depend on low-probability, high-severity events, the contractual trigger, the timing of coupons, and the term structure of interest rates. The market motivation is large. Industry sources report that the outstanding CAT bond market expanded from about \$1 billion to \$49.1 billion by June 2024, while indemnity and index triggers remain the dominant contractual forms \citep{Artemis2024, BraunKousky2021}. Indemnity triggers are particularly important in practice because they protect the sponsor against its own insured loss; indemnity contracts account for roughly 70\% of trigger types, industry-index contracts for 20\%, parametric triggers for 3\%, modeled-loss triggers for 1\%, and other structures for 6\%, see \cite{BraunKousky2021}. {An indemnity trigger is based on the sponsor's own covered loss rather than an external industry index or purely physical event measure; it is therefore closely aligned with hedging needs but requires a model of the sponsor-specific aggregate loss process.} This heterogeneity matters for valuation because the trigger definition determines how catastrophe losses are translated into investor cash-flow losses.

Structural models \cite{BaryshnikovMayoTaylor2001,BurneckiKukla2003,Hrdle2007CalibratingCB,hofer2021risk,
Jarrow2010,reshetar2008pricing,siyamah2021cat}  are attractive because they connect prices to interpretable inputs such as catastrophe intensity, severity distribution, maturity, interest rates, coupon frequency, and attachment threshold. Their drawback is speed. A single zero-coupon CAT bond price requires a rare-event probability. A coupon bond requires the same calculation at several cash-flow dates. A trading, issuance, or risk-management workflow may need thousands of prices under changing assumptions, for example to screen structures, update marks, or compute stress sensitivities. Monte Carlo and partial integro-differential equation (PIDE) methods (see, e.g., \cite{BaryshnikovMayoTaylor2001}) can be accurate, but they are not designed for interactive use when every proposed contract requires a fresh tail-probability calculation.

Neural surrogates are unusually valuable in this asset class because CAT bond pricing is an amortized rare-event valuation problem. For liquid equities, vanilla options, or government bonds, practitioners often observe dense quotes or use highly standardized pricing routines often involving closed-form valuation formulas; the marginal computational cost of another price is modest. A CAT bond is different. Each issue is comparatively bespoke, and each structural price depends on whether aggregate losses cross a contractual attachment point, often far in the tail. Thus changing the threshold, maturity, coupon frequency, or severity assumptions changes the rare-event probability itself. The expensive object is not a single price but a pricing surface over contract and catastrophe-risk inputs. Training the neural network on variance-reduced simulations pays this cost once and then reuses the learned surface for valuation, structure search, and sensitivity analysis.

In this paper we learn the structural pricing map once and then evaluate it quickly. The contribution is not an empirical spread-prediction model and not a new catastrophe-loss model. Instead, the network approximates the valuation operator implied by a chosen compound-Poisson model. This places the paper in the emerging deep-surrogate literature in finance; \cite{ChenDidisheimScheidegger2026, HorvathMuguruzaTomas2021, HutchinsonLoPoggio1994, LiuOosterleeBohte2019, NeufeldSester2022}. Our setting differs because the central numerical object is a catastrophe-trigger probability rather than an option-implied volatility surface. The labels are produced by a Monte Carlo engine with importance sampling, following the rare-event logic of \citet{Glasserman2004}. The trained network therefore preserves the economics of the structural model while replacing repeated simulation by a fast function evaluation.

{In plain language, the method trades an expensive one-time simulation and training step for extremely cheap repeated pricing. This is useful when many contracts or stress scenarios must be evaluated under the same model. The limitations are also important: the surrogate is conditional on the chosen pricing measure, severity law, trigger specification, and simulated input domain. In particular, new model specifications or out-of-domain contracts require computationally expensive additional simulation and retraining rather than blind extrapolation.} 

The paper focuses on accuracy, speed, and risk sensitivities while the supplementary material gives the technical justification: the neural approximation result is based on the universal approximation theorem \cite{HornikStinchcombeWhite1989, pinkus1999approximation}, and the importance-sampling construction is derived together with a mathematical variance-reduction argument. The supplementary material moreover provides additional details on the numerical experiments and the used algorithms. Code and notebooks are available at {\url{https://github.com/HuansangXu/CAT-bonds}}.

\section{Pricing map and learning design}
For our pricing approach, we assume catastrophe arrivals follow a Poisson process $M$ with constant intensity $\lambda$ under the pricing measure. Severities $(X_i)_{i\geq1}$ are independent and identically distributed and independent of $M$. Aggregate loss at time $T$ is
\begin{equation}
L(T)=\sum_{i=1}^{M(T)}X_i,
\end{equation}
and the trigger occurs when $L(T)\geq D$ for some threshold $D$. In line with \cite{cuchiero2010affine, Duffie1996AYM, duffee2002term, keller2011affine, piazzesi2010affine} discounting is generated by an affine short-rate model; the numerical study uses the Vasicek model (\cite{Vasicek1977}) with mean reversion 0.2, long-run rate 3\%, and volatility 2\%. For a zero-coupon CAT bond with face value $F$, the normalized structural price is
\begin{equation}\label{eq:price1}
V_Z(F,T;r_0,\lambda,D)=P_Z(F,0,T)\{1-\theta(T;\lambda,D)\},\quad
\theta(T;\lambda,D)=\mathbb{Q}(L(T)\geq D),
\end{equation}
where $P_Z$ is the default-free zero-coupon price and $r_0$ the initial short-rate. A coupon bond is a finite sum of protected coupon and principal claims denoted by $\{C_j\}_{j=1,\dots,N}$ and $F$, respectively:
\begin{equation}\label{eq:price2}
V_C=V_Z(F,T;r_0,\lambda,D)+\sum_{j=1}^{N}V_Z(C_j,t_j;r_0,\lambda,D).
\end{equation}
The representations in \eqref{eq:price1} and \eqref{eq:price2} isolate the numerical bottleneck: the repeated estimation of the tail probability $\theta$. The decomposition also explains why a surrogate is useful. The default-free discount factor is available in closed form under the affine term-structure model, while the catastrophe component is a path-dependent tail probability. In coupon structures, the tail probability must be recomputed for each protected coupon date as well as for principal. Therefore, even if each individual simulation is well optimized, the full valuation task scales with the number of cash flows and with the number of candidate structures considered by the user.

Plain Monte Carlo estimates $\theta$ by simulating $L(T)$ and averaging the indicator $I\{L(T)\geq D\}$, whose variance is $(\theta-\theta^2)/n$. When $D$ lies far in the tail, most simulated paths do not trigger, so the estimator converges slowly. The training labels therefore use importance sampling. Under the sampling density $g$, simulated losses are reweighted by the likelihood ratio $R=f/g$, giving the unbiased estimator $n^{-1}\sum_i I\{L_i(T)\geq D\}R(L_i(T))$. For Gamma severities, exponential tilting changes both the Poisson count and the severity scale; the tilting parameters are chosen so that the tilted expected aggregate loss is close to the trigger threshold. For Lognormal severities, where the moment generating function is infinite for positive arguments, the log-severity mean is shifted and the likelihood ratio corrects the change of measure. A diagnostic run with 20,000 simulations, $D=9\times10^9$, $T=1$, and $\lambda=35$ shows faster stabilization of the trigger-probability estimate under importance sampling for both severity laws. 

The neural network maps the contract and market state $(r_0,\lambda,D,T,N)$ to normalized CAT bond prices. Separate networks are trained for Gamma and Lognormal severities. {
These two severity laws are used because they are common positive loss-size specifications in CAT-bond applications, see \citep{BurneckiKukla2003,Sukono2022}. This comparison should be interpreted as a robustness check across two standard severity laws, not as a full treatment of severity-model uncertainty. Broader severity families, tail calibration, and event-data-based model selection are left for future work.} Each selected architecture has hidden layers $(256,128,64,32)$, ReLU activation, batch normalization, dropout 0.1, $L_2$ regularization $10^{-4}$, and learning rate $10^{-5}$. Hyperparameters are selected using five-fold cross-validation with mean squared error as the criterion, and the final evaluation uses an 80/20 train-test split. Table~\ref{tab:design} summarizes the simulation design.

\begin{table}[t]
\centering
\caption{Simulation design for structural training labels}
\label{tab:design}
\begin{tabular}{ll}
\toprule
Input or parameter & Baseline design \\
\midrule
Coupon rate & 5\% \\
Vasicek parameters & mean reversion 0.2, long-run rate 3\%, volatility 2\% \\
Initial short rate & $r_0\sim U(0,0.08)$ \\
Intensity & $\lambda\sim U(30,40)$ \\
Threshold & $D\sim U(7,13)$ billion dollars \\
Maturity & $T\sim U(90,720)$ days \\
Coupon payments & $N\in\{0,2,3,4,6,8,10,12\}$ \\
Severity laws & Gamma$(1,1.635\times10^8)$; Lognormal$(18.4,1)$ \\
Training labels & 600,000 simulated prices for each severity law \\
\bottomrule
\end{tabular}
\end{table}

{The ranges in Table~\ref{tab:design} define the operational input domain of the neural surrogate. Hence the network should be interpreted as a fast interpolator on this compact domain, not as a validated extrapolation rule. If a contract or stress scenario falls outside the stated ranges, the simulation design should be enlarged and the surrogate retrained, or direct MC-IS should be used for that case. Boundary diagnostics are reported in the supplementary material.}

\section{Accuracy and computational gains}
Figure~\ref{fig:accuracy} reports out-of-sample accuracy. For both severity laws, predicted prices lie close to the simulated structural prices and the error distributions are narrow and centered around zero. The average mean absolute error is 0.00274 for Gamma severities and 0.00308 for Lognormal severities; the corresponding mean squared errors are $1.4\times10^{-5}$ and $1.8\times10^{-5}$, i.e. RMSEs of 0.00380 and 0.00429. {The supplementary material reports additional error statistics, including signed bias, 95\% and 99\% absolute-error quantiles, maximum absolute error, $R^2$, and boundary-region errors.} Since prices are normalized, these errors are small relative to the variation generated by maturity, threshold, coupon frequency, and severity assumptions. The diagnostic is therefore not only statistical but economic: the surrogate recovers high and low structural values without visible distortion across the price range.

\begin{figure}[t]
\centering
\begin{subfigure}{0.48\textwidth}
\includegraphics[width=\linewidth]{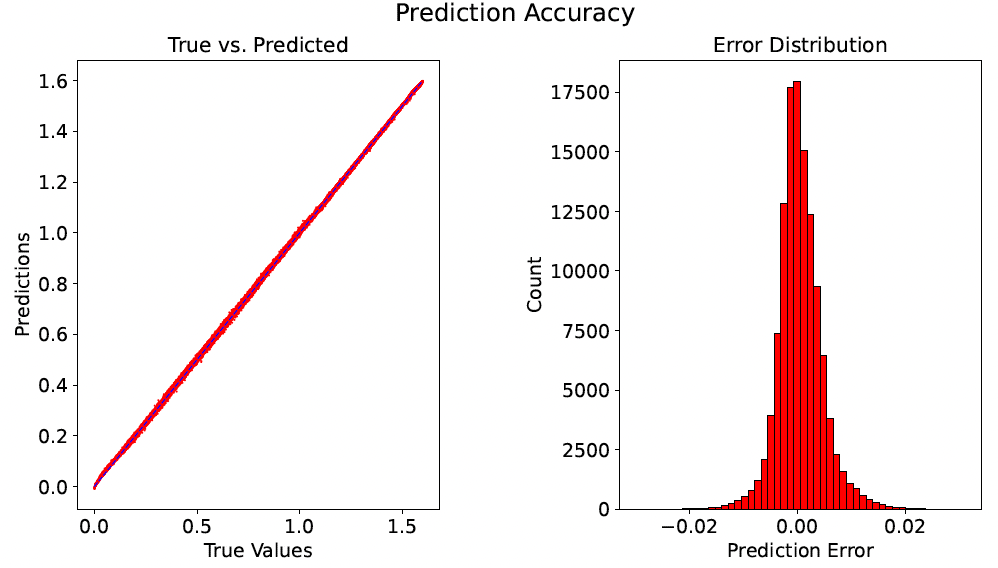}
\caption{Gamma severities}
\end{subfigure}\hfill
\begin{subfigure}{0.48\textwidth}
\includegraphics[width=\linewidth]{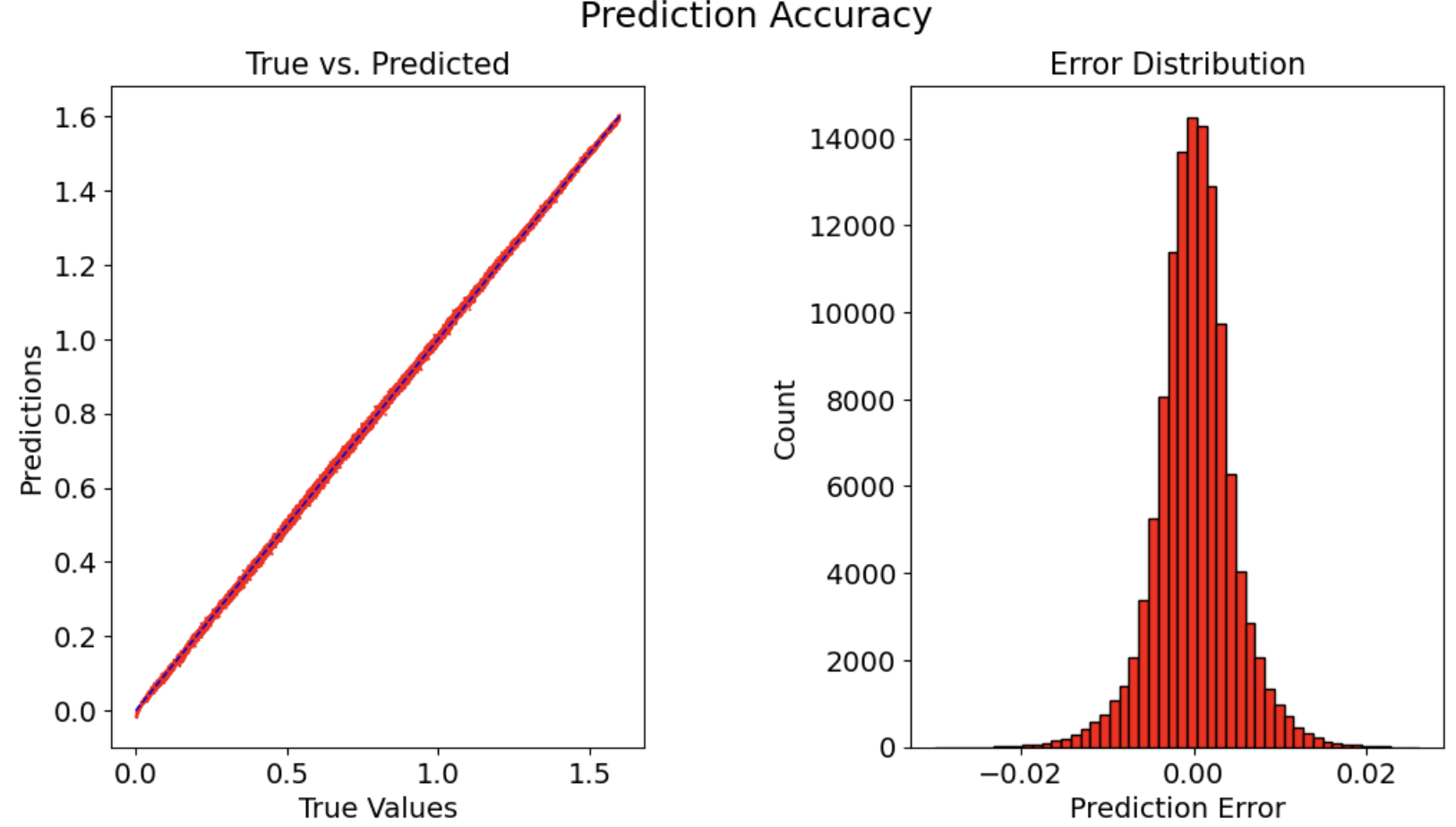}
\caption{Lognormal severities}
\end{subfigure}
\caption{Out-of-sample pricing accuracy. Predicted values align with simulated structural prices and prediction errors are concentrated around zero.}
\label{fig:accuracy}
\end{figure}

Table~\ref{tab:timing} compares PIDE prices, Monte Carlo with importance sampling, and the neural surrogate for 1000 prices, fixing $r_0=3\%$, $\lambda=35$, and $D=9\times10^9$. {Values are close across methods in most benchmark scenarios. The largest deviations occur in the long-maturity, high-coupon Gamma cases, where the PIDE benchmark is more conservative than both MC-IS and the neural surrogate. The neural surrogate remains close to the MC-IS labels on which it is trained.} For instance, in the $N=4,T=1$ Gamma case the PIDE, MC-IS, and neural prices are 1.1549, 1.1518, and 1.1510; for Lognormal severities they are 1.1397, 1.1361, and 1.1357.

\begin{table}[H]
\centering
\small
\caption{Value and time comparison for 1000 prices}
\label{tab:timing}
\begin{tabular}{llrrrrrr}
\toprule
Scenario & Severity & \multicolumn{3}{c}{Value} & \multicolumn{3}{c}{Time} \\
\cmidrule(lr){3-5}\cmidrule(lr){6-8}
 & & PIDE & MC-IS & NN & PIDE & MC-IS & NN \\
\midrule
$N=0,T=1$ & Gamma & 0.9591 & 0.9563 & 0.9533 & 15.00 h & 39.61 s & 0.0350 s \\
$N=0,T=1$ & Lognormal & 0.9449 & 0.9414 & 0.9428 & 15.33 h & 53.65 s & 0.0318 s \\
$N=2,T=1$ & Gamma & 1.0564 & 1.0533 & 1.0533 & 43.94 h & 118.91 s & 0.0312 s \\
$N=2,T=1$ & Lognormal & 1.0413 & 1.0377 & 1.0392 & 46.47 h & 159.38 s & 0.0298 s \\
$N=4,T=1$ & Gamma & 1.1549 & 1.1518 & 1.1510 & 73.46 h & 220.73 s & 0.0430 s \\
$N=4,T=1$ & Lognormal & 1.1397 & 1.1361 & 1.1357 & 78.80 h & 261.32 s & 0.0294 s \\
$N=8,T=2$ & Gamma & 0.3146 & 0.3783 & 0.3823 & 131.34 h & 342.13 s & 0.0329 s \\
$N=8,T=2$ & Lognormal & 0.4354 & 0.4257 & 0.4202 & 142.29 h & 408.17 s & 0.0298 s \\
$N=12,T=2$ & Gamma & 0.4678 & 0.5331 & 0.5310 & 190.16 h & 438.21 s & 0.0311 s \\
$N=12,T=2$ & Lognormal & 0.5931 & 0.5822 & 0.5773 & 205.85 h & 581.54 s & 0.0301 s \\
\bottomrule
\end{tabular}

{\begin{minipage}{0.98\linewidth}\footnotesize\vspace{0.3em}
Notes: PIDE denotes the partial integro-differential equation benchmark, MC-IS denotes Monte Carlo with importance sampling, and NN denotes the trained neural network. Prices are normalized by face value. Note that values above one can occur because protected coupon payments are included. The MC-IS benchmark uses 5,000 importance-sampling paths per protected cash-flow date with a fixed-path Monte Carlo simulation scheme.
%The MC-IS benchmark uses {\color{red}[number of simulations]} paths per protected cash-flow date with target Monte Carlo standard error {\color{red}[target standard error]} and convergence rule {\color{red}[stopping/convergence criterion]}.
\end{minipage}}

\end{table}

The timing comparison shows the true value of the neural surrogate approach. Neural evaluation takes only 0.0294--0.0430 seconds for 1000 prices across all scenarios, while MC-IS requires 39.61--581.54 seconds and the PIDE benchmark 15.00--205.85 hours. At $N=12,T=2$, MC-IS is roughly 14,000 times slower than the neural surrogate for Gamma severities and about 19,000 times slower for Lognormal severities. The PIDE method is slower still. These differences change the feasible workflow: structural CAT bond pricing can be used inside screening, calibration, and stress-testing loops rather than only after a candidate structure has been selected. {These timings are online valuation times after training. The one-time offline costs are label generation ({24 hours 21 minutes} for Gamma and { 33 hours 20 minutes} for Lognormal severities), training of the selected architecture ({ 29.70 and 35.98 minutes for Gamma and Lognormal severities, respectively}), and full cross-validation/model selection ({7 hours 27 minutes}) on {an Apple M3 Pro machine with 11 CPU cores, 14 GPU cores, and 18 GB of unified RAM}. These costs are amortized only when the pricing surface is reused; for an isolated out-of-domain valuation, direct MC-IS can be preferable. Full timing details by architecture are reported in the supplementary material.}

{The supplementary material further reports several robustness diagnostics, including boundary-region errors near the edge of the training domain, monotonicity violation rates on dense grids, and the performance of alternative network architectures. The boundary-region analysis indicates that the neural network retains comparable accuracy near the edges of the training domain. Moreover, on dense grids with $100{,}001$ points for each parameter, including values outside the training domain, the model exhibits no monotonicity violations for either severity distribution; see also Section~\ref{sec:Sens}. Finally, the procedure is stable with respect to the network architecture: alternative architectures yield comparable prediction accuracy, while the proposed architecture $(256,128,64,32)$ attains the lowest MAE with a reasonable training time.}

\section{Sensitivities}\label{sec:Sens}
The trained pricing map also gives immediate sensitivity curves. Figures~\ref{fig:gamma_sens} and \ref{fig:log_sens} vary catastrophe intensity, attachment threshold, and the initial short rate under Gamma and Lognormal severities. These curves are evaluations of the same surrogate, not separate models.

The results match the economics of CAT bond risk. Prices decline with catastrophe intensity because more frequent arrivals raise the probability of principal impairment; the effect is steepest at low attachment thresholds and much flatter when $D=11$ billion. Prices move in the opposite direction with the threshold: increasing $D$ makes the trigger harder to reach, with the strongest marginal effect around 8--9 billion dollars where the trigger probability changes most rapidly. Prices also decline with the initial short rate $r_0$, consistent with higher discount rates lowering present values, while the shape of this rate response is only weakly affected by $\lambda$. The figures therefore serve both as validation and as a risk-management output. They show economically correct monotonicities, identify nonlinear regions where contract design matters most, and provide full response curves without recomputing rare-event simulations on a grid. {This analysis directly links the speed gain of the presented method to a risk-management use case: each curve requires many valuations, which would otherwise require repeated rare-event simulations. A finite-difference monotonicity check on dense grids is reported in the supplementary material, with violation rates. The results show 0\% monotonicity violation rates for $\lambda$, $D$, and $r_0$ under both the Gamma and Lognormal severity models, confirming that the neural network preserves the expected monotonic relationships across the parameter domain.}

\begin{figure}[!htbp]
\centering
\begin{subfigure}{0.32\textwidth}\includegraphics[width=\linewidth]{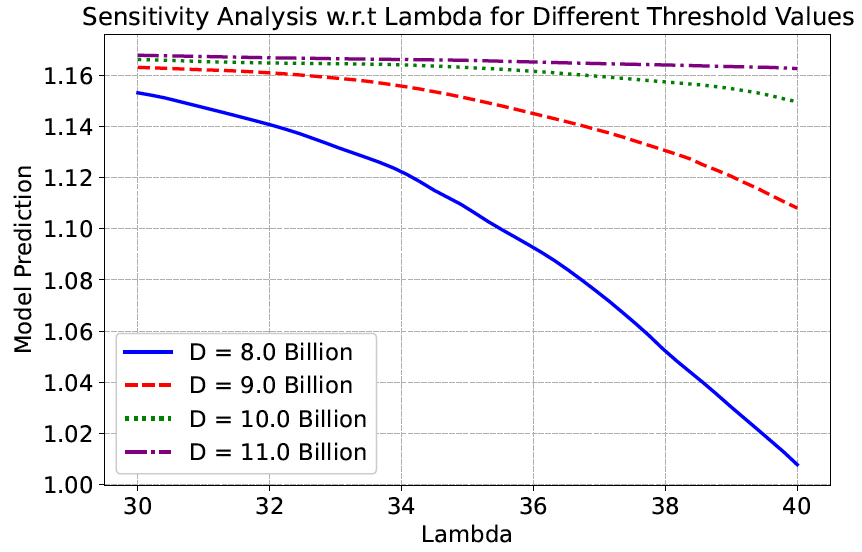}\caption{Intensity by threshold}\end{subfigure}\hfill
\begin{subfigure}{0.32\textwidth}\includegraphics[width=\linewidth]{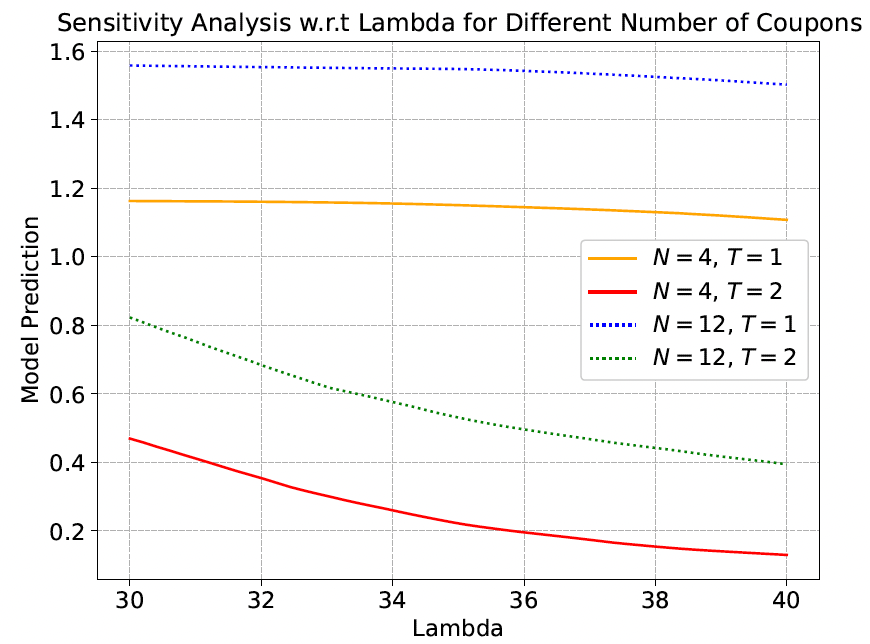}\caption{Intensity by coupons}\end{subfigure}\hfill
\begin{subfigure}{0.32\textwidth}\includegraphics[width=\linewidth]{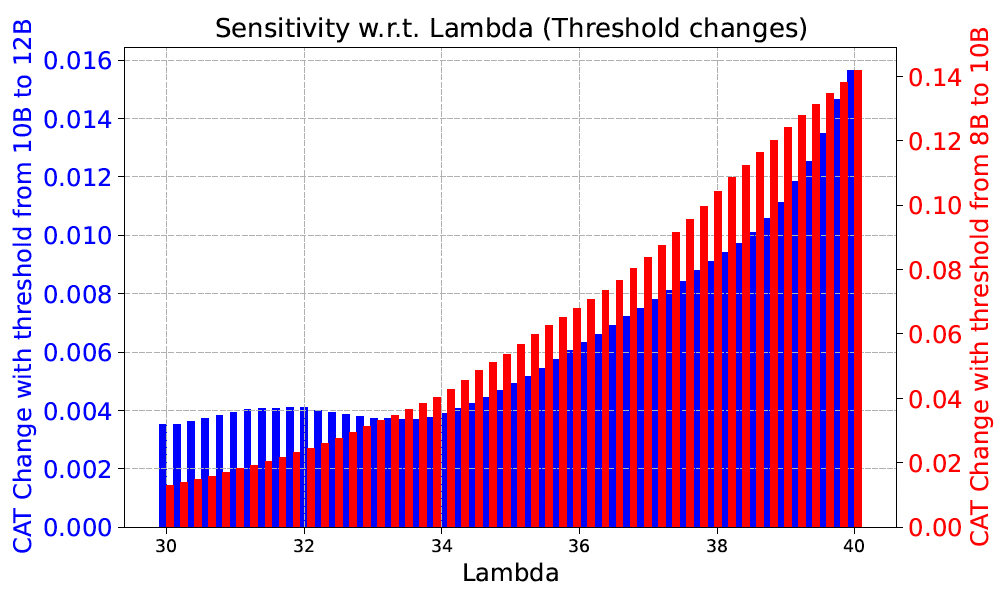}\caption{Intensity changes}\end{subfigure}
\begin{subfigure}{0.32\textwidth}\includegraphics[width=\linewidth]{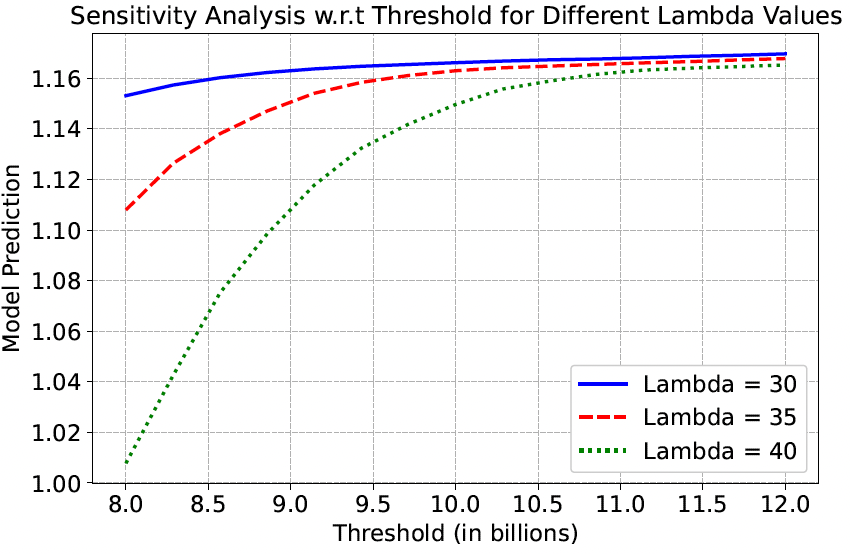}\caption{Threshold by intensity}\end{subfigure}\hfill
\begin{subfigure}{0.32\textwidth}\includegraphics[width=\linewidth]{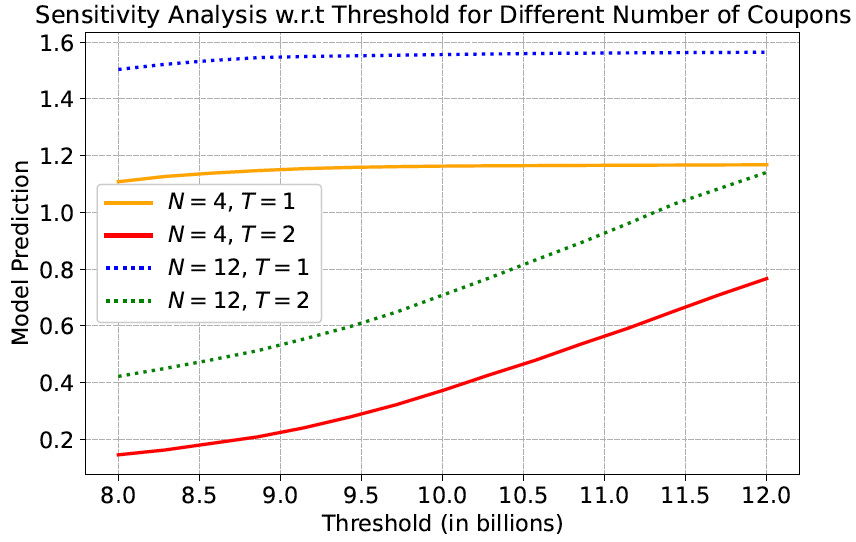}\caption{Threshold by coupons}\end{subfigure}\hfill
\begin{subfigure}{0.32\textwidth}\includegraphics[width=\linewidth]{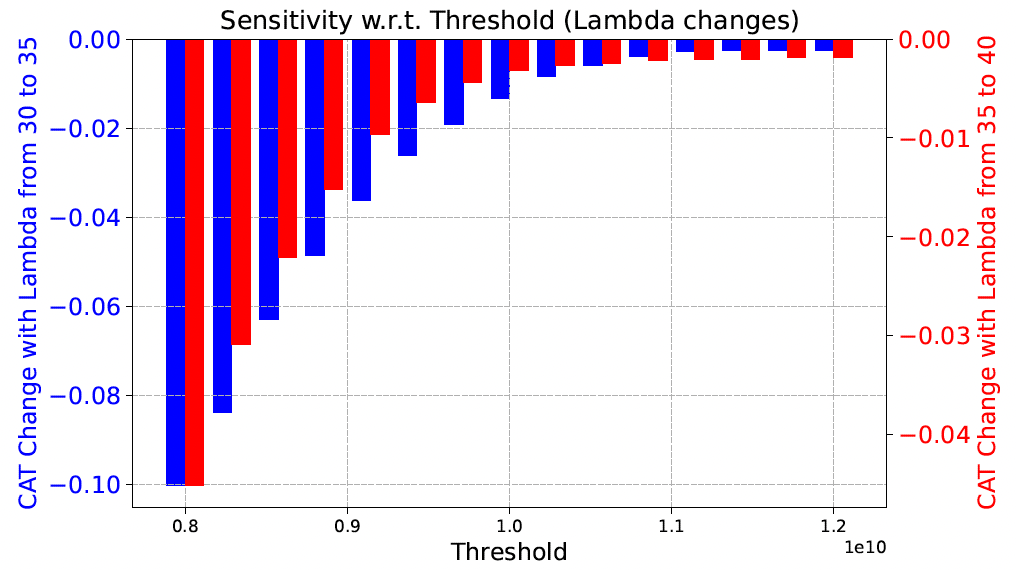}\caption{Threshold changes}\end{subfigure}
\begin{subfigure}{0.32\textwidth}\includegraphics[width=\linewidth]{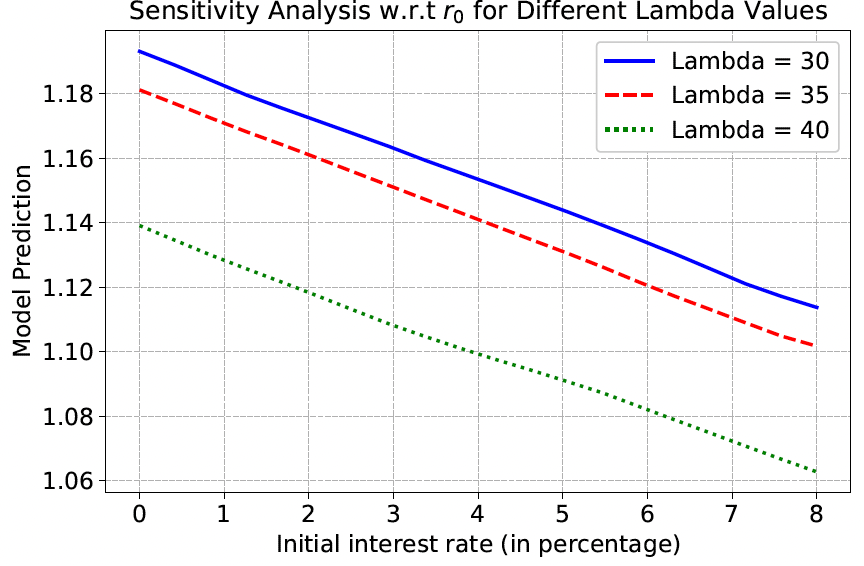}\caption{Rate by intensity}\end{subfigure}\hfill
\begin{subfigure}{0.32\textwidth}\includegraphics[width=\linewidth]{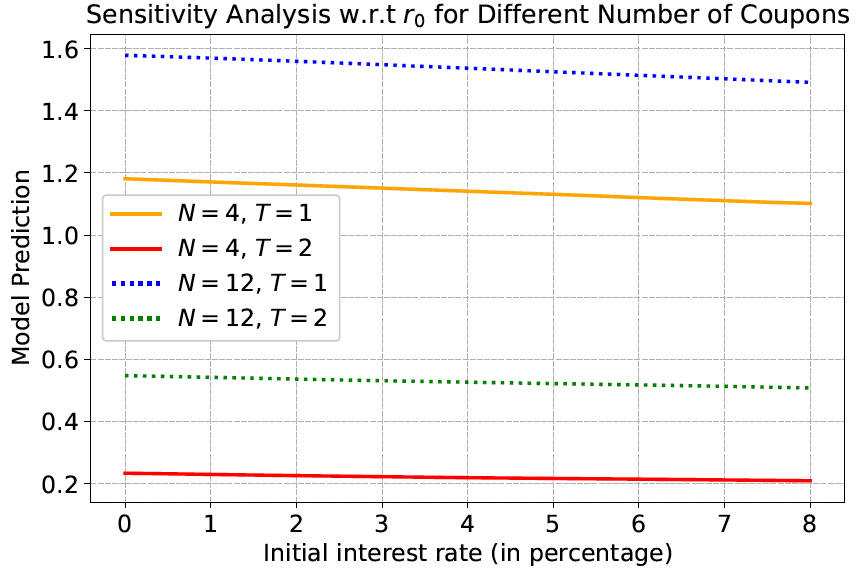}\caption{Rate by coupons}\end{subfigure}\hfill
\begin{subfigure}{0.32\textwidth}\includegraphics[width=\linewidth]{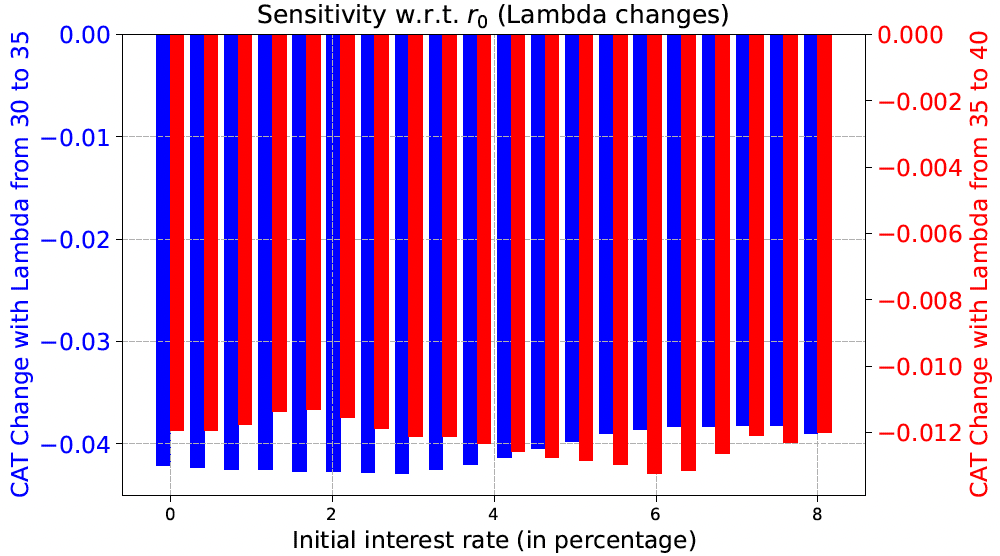}\caption{Rate changes}\end{subfigure}
\caption{Computed price sensitivities for Gamma distributed losses. The panels vary catastrophe intensity $\lambda$, threshold $D$, and initial short rate $r_0$ under alternative threshold, coupon, and intensity scenarios.}
\label{fig:gamma_sens}
\end{figure}

\begin{figure}[!htbp]
\centering
\begin{subfigure}{0.32\textwidth}\includegraphics[width=\linewidth]{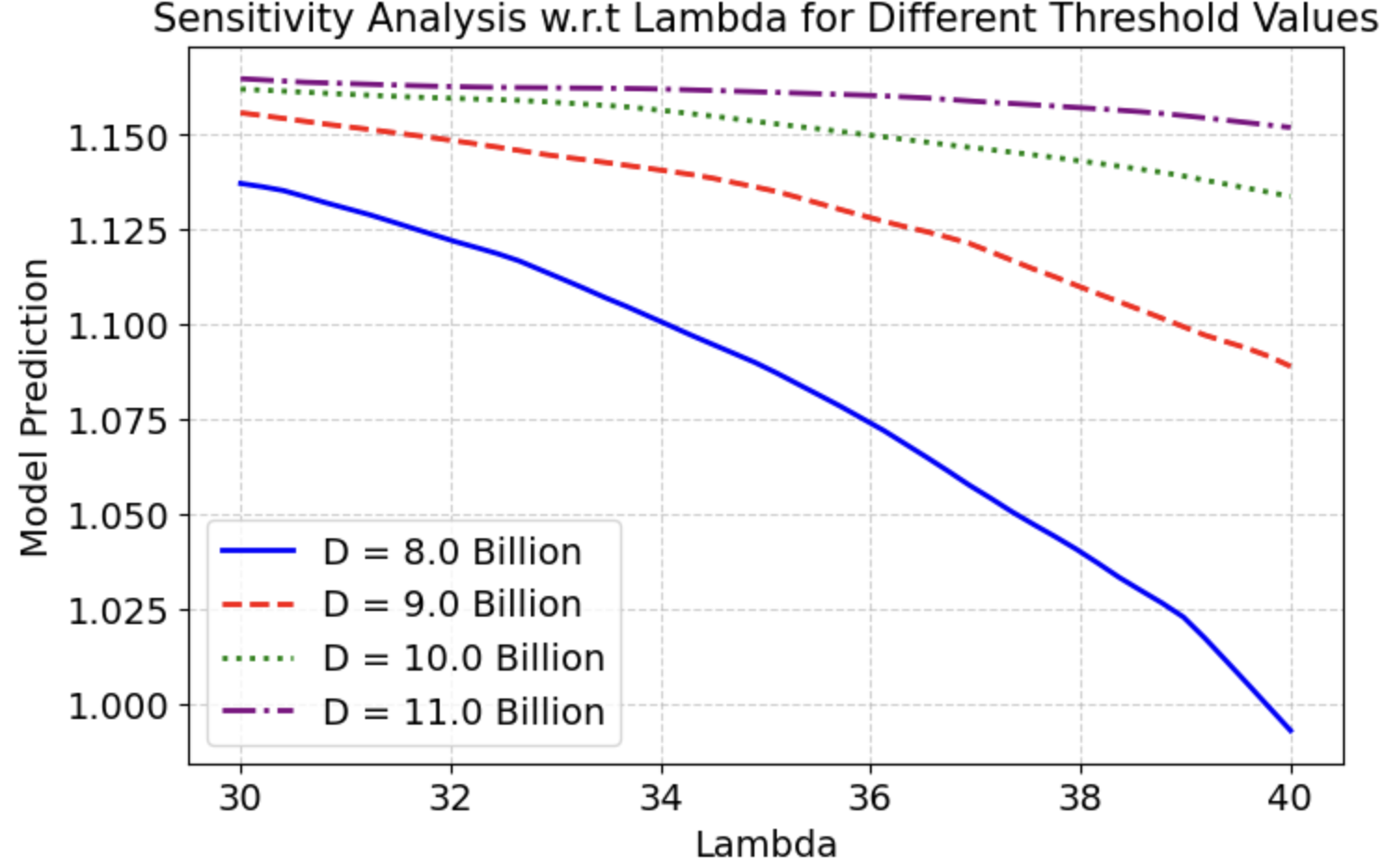}\caption{Intensity by threshold}\end{subfigure}\hfill
\begin{subfigure}{0.32\textwidth}\includegraphics[width=\linewidth]{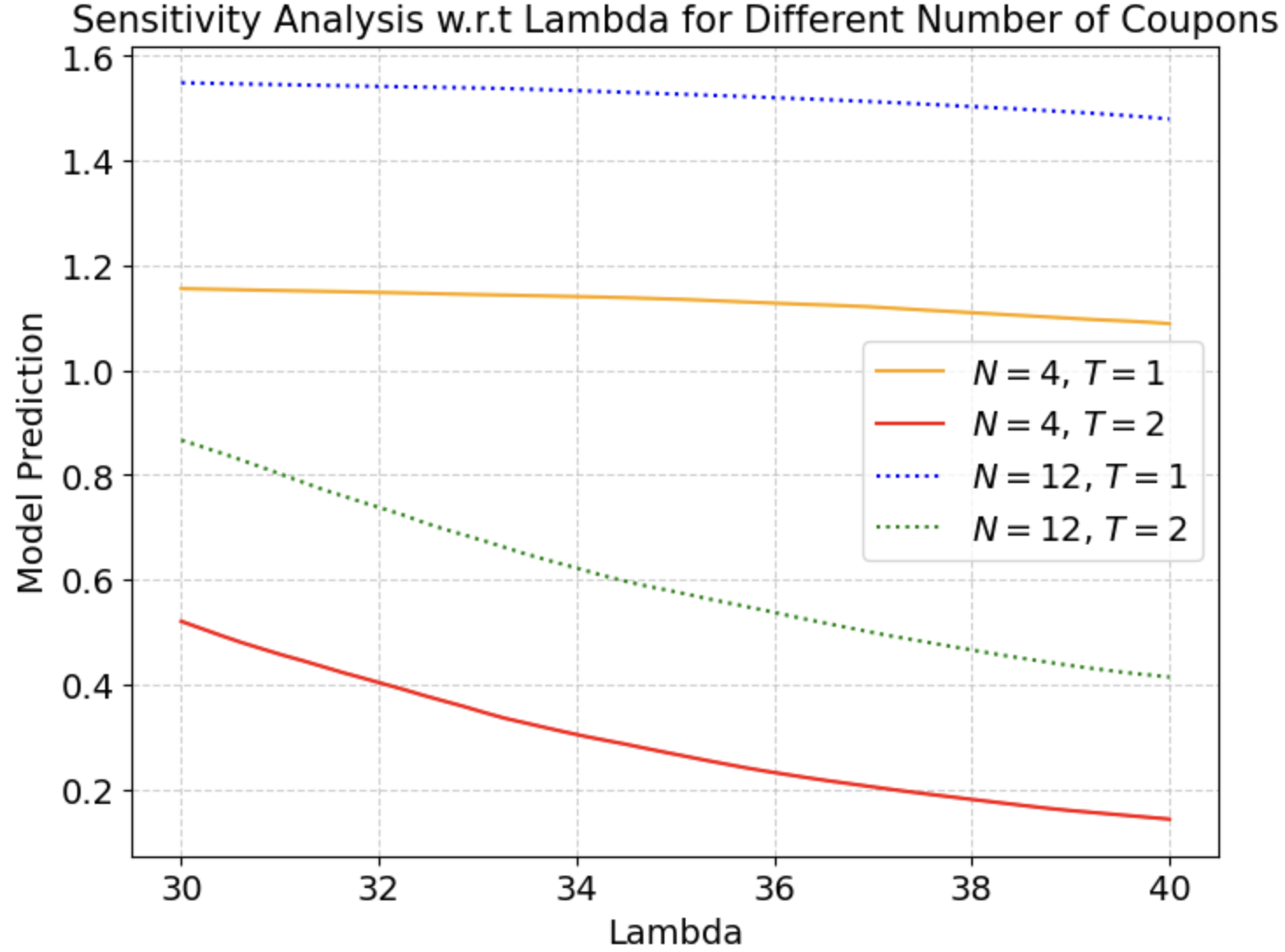}\caption{Intensity by coupons}\end{subfigure}\hfill
\begin{subfigure}{0.32\textwidth}\includegraphics[width=\linewidth]{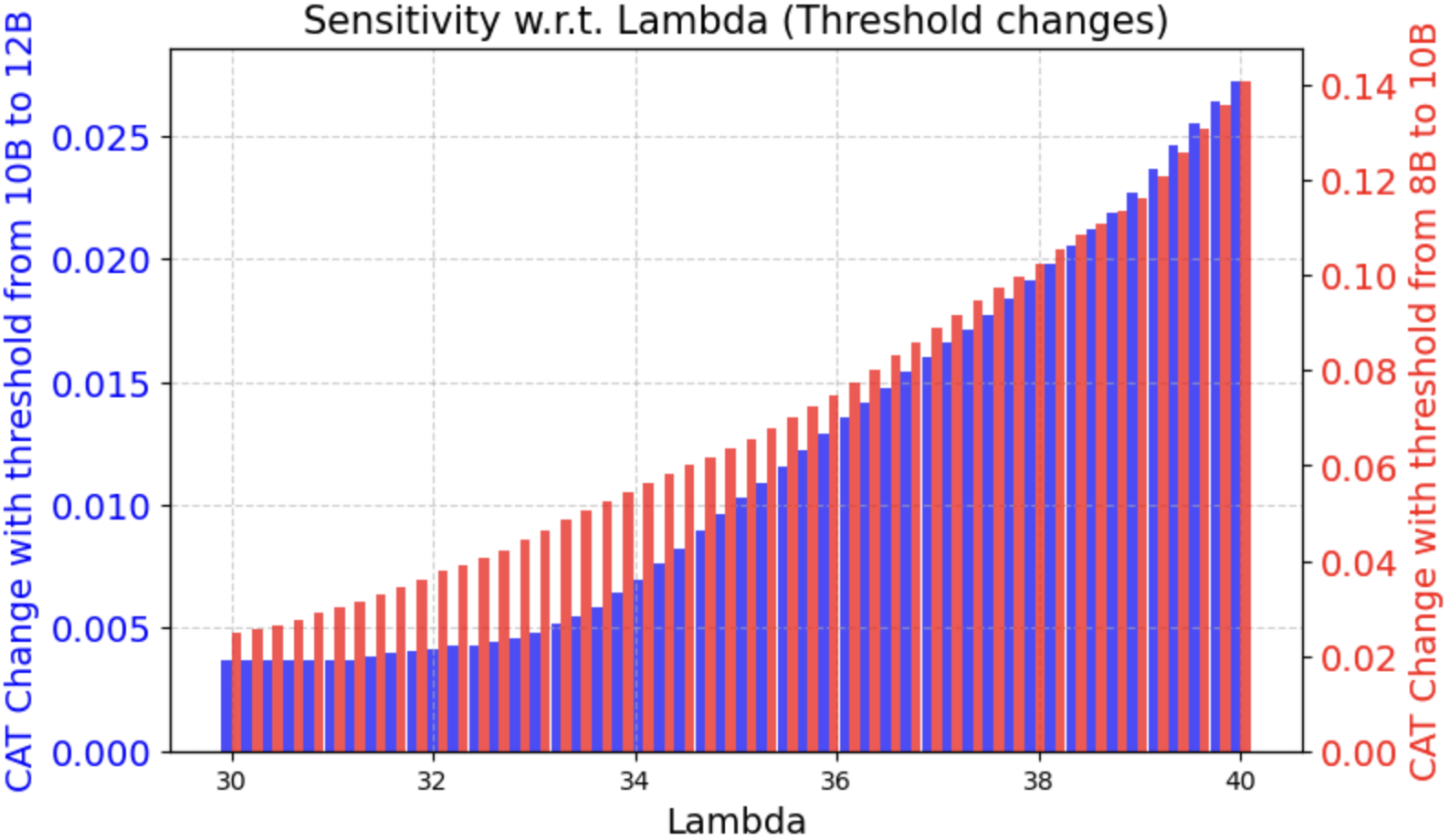}\caption{Intensity changes}\end{subfigure}
\begin{subfigure}{0.32\textwidth}\includegraphics[width=\linewidth]{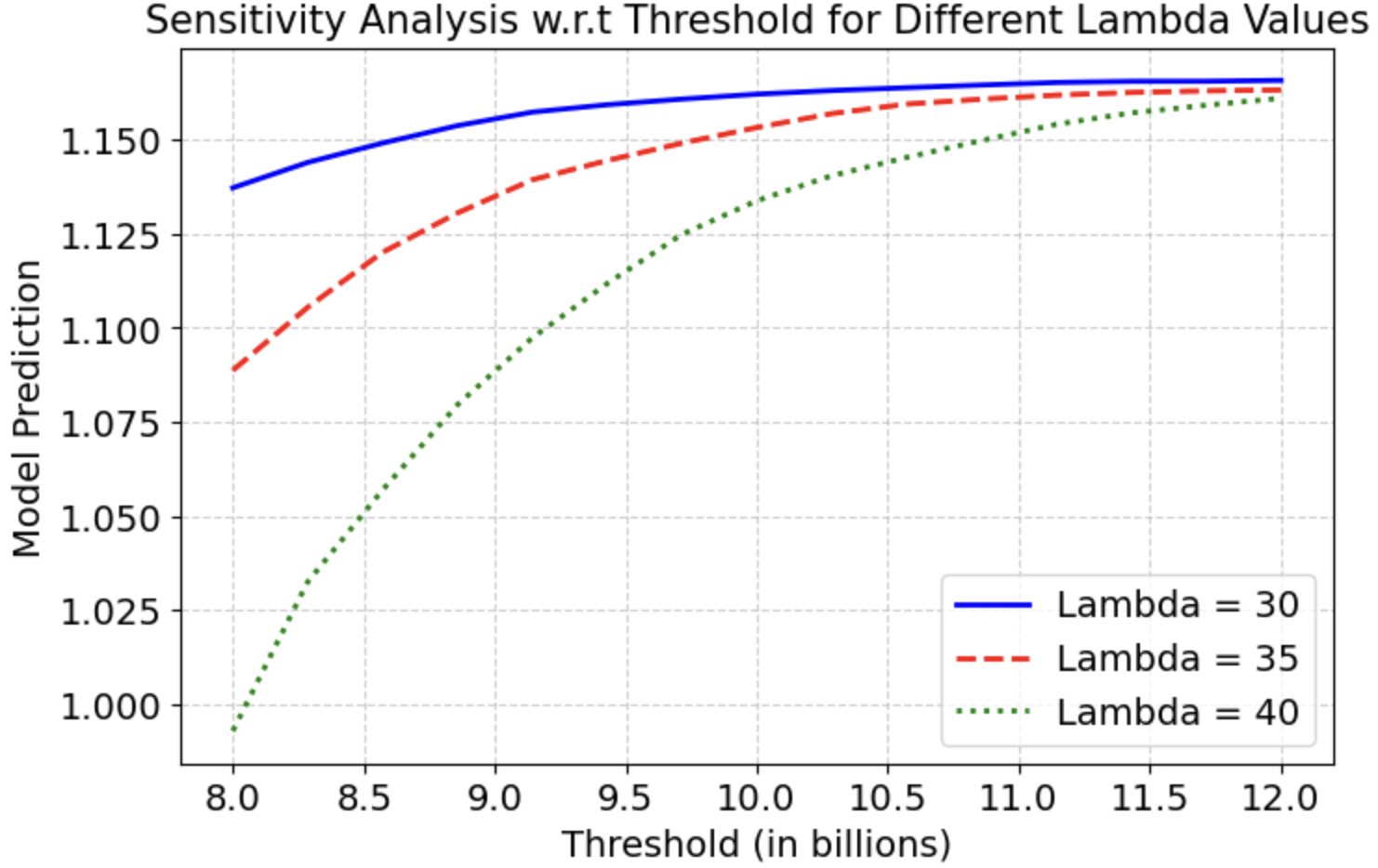}\caption{Threshold by intensity}\end{subfigure}\hfill
\begin{subfigure}{0.32\textwidth}\includegraphics[width=\linewidth]{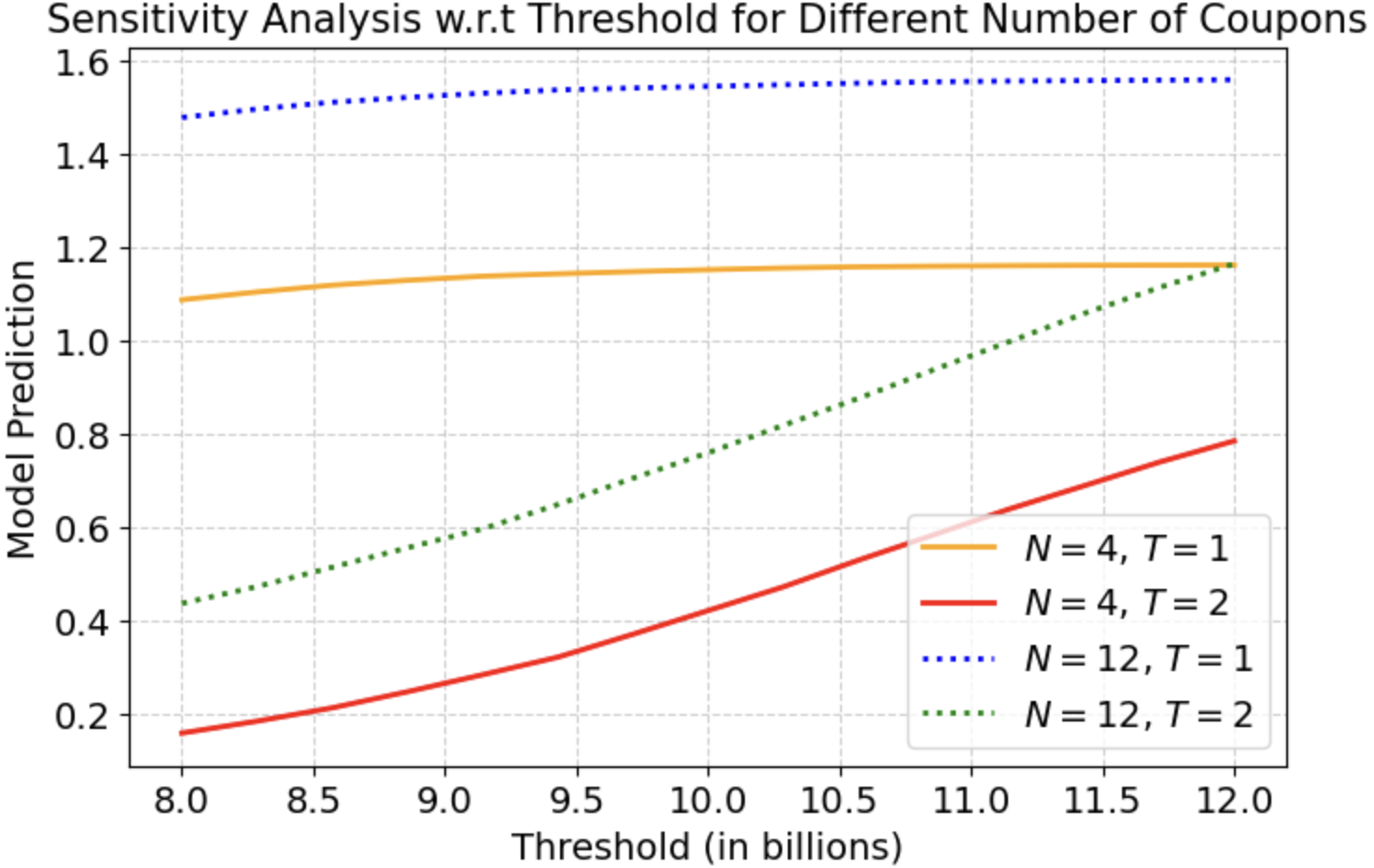}\caption{Threshold by coupons}\end{subfigure}\hfill
\begin{subfigure}{0.32\textwidth}\includegraphics[width=\linewidth]{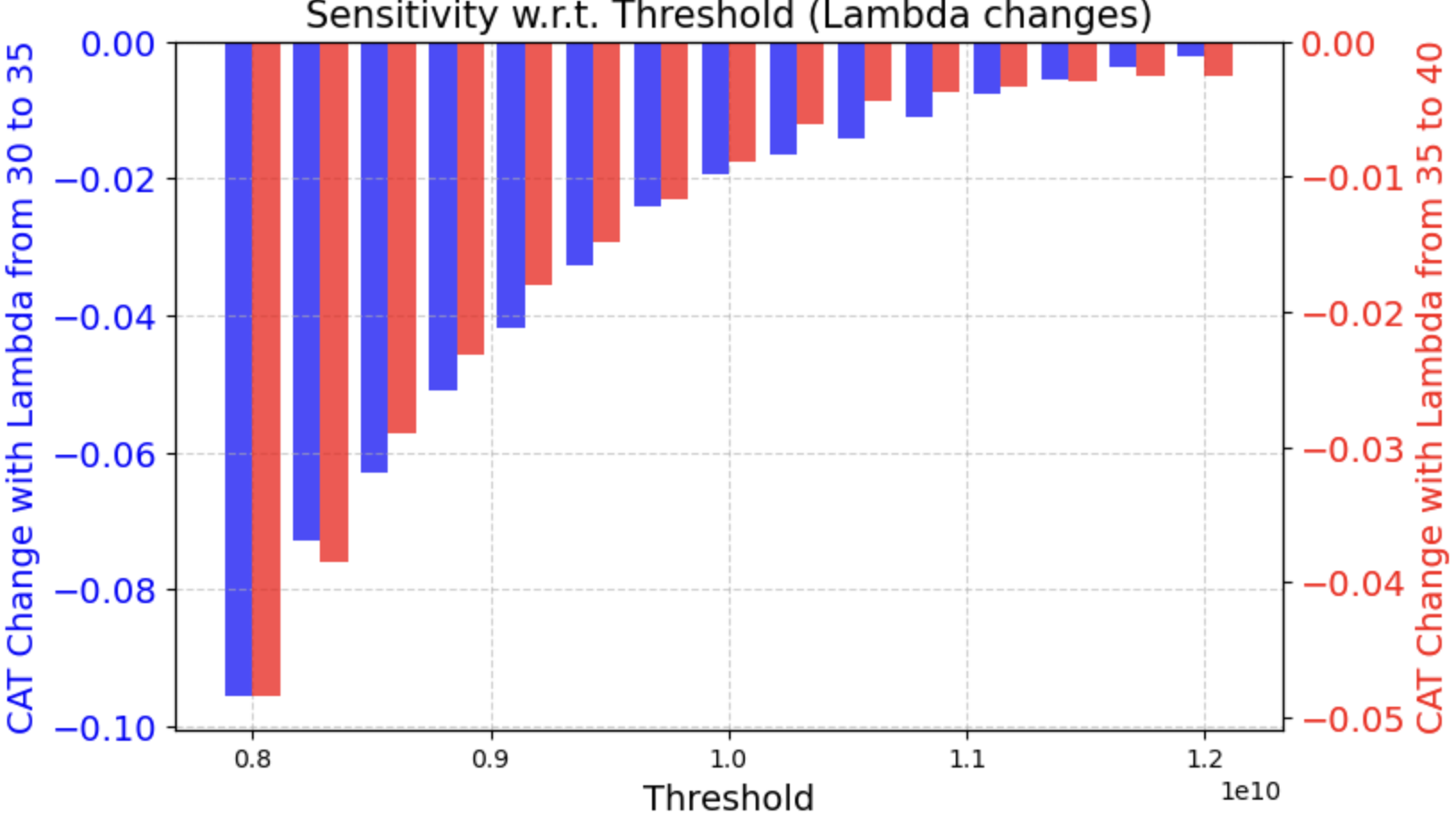}\caption{Threshold changes}\end{subfigure}
\begin{subfigure}{0.32\textwidth}\includegraphics[width=\linewidth]{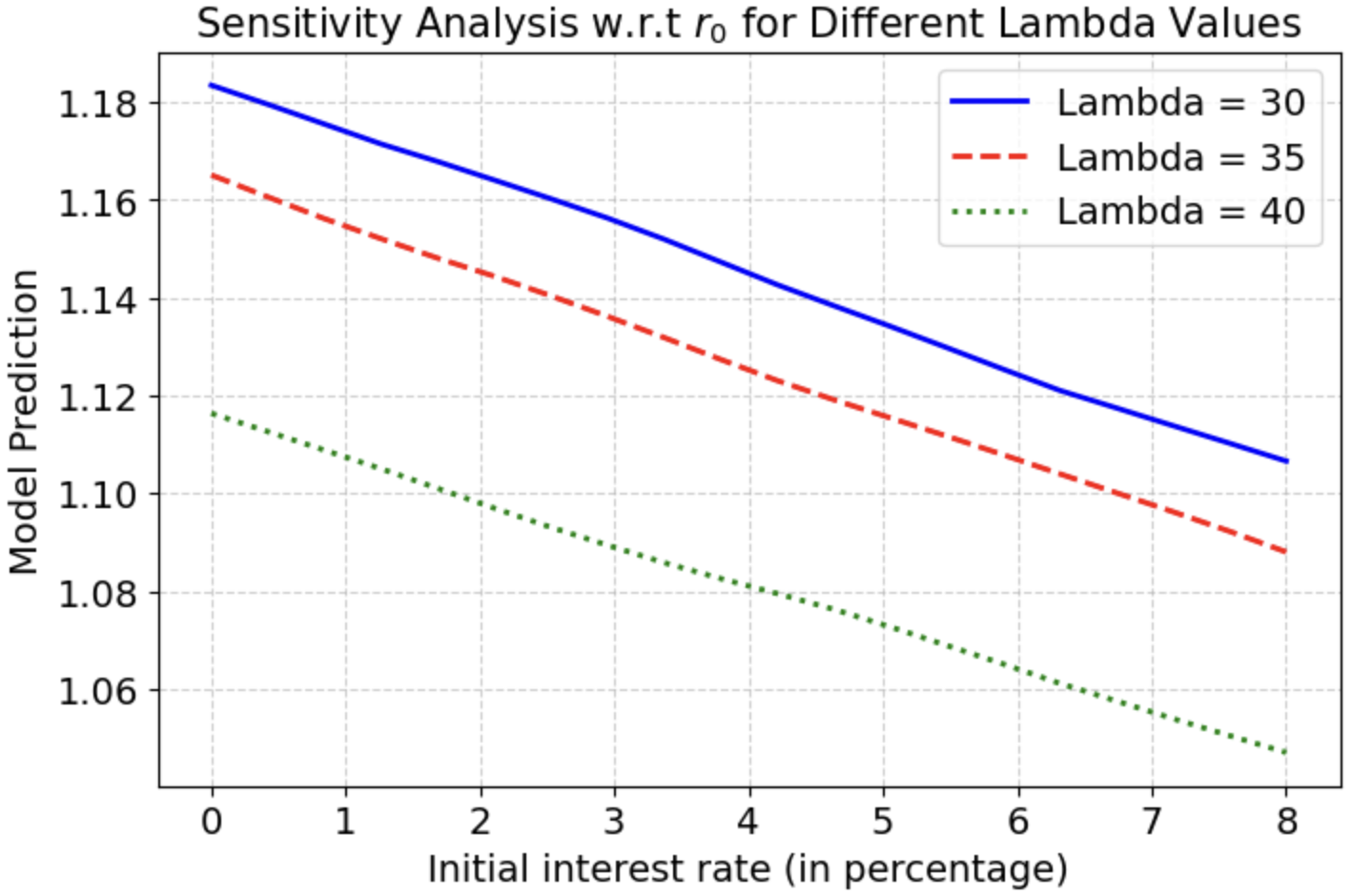}\caption{Rate by intensity}\end{subfigure}\hfill
\begin{subfigure}{0.32\textwidth}\includegraphics[width=\linewidth]{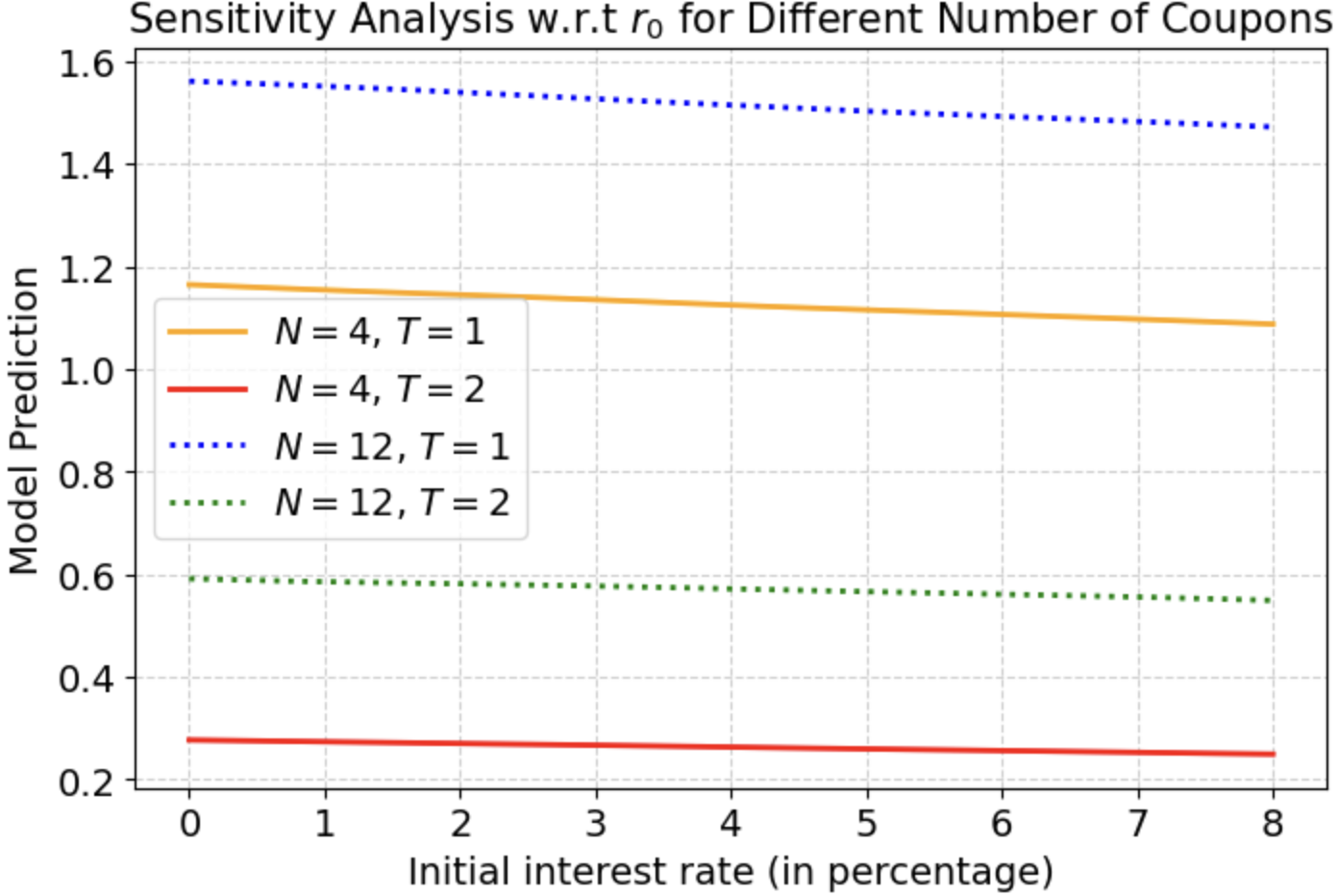}\caption{Rate by coupons}\end{subfigure}\hfill
\begin{subfigure}{0.32\textwidth}\includegraphics[width=\linewidth]{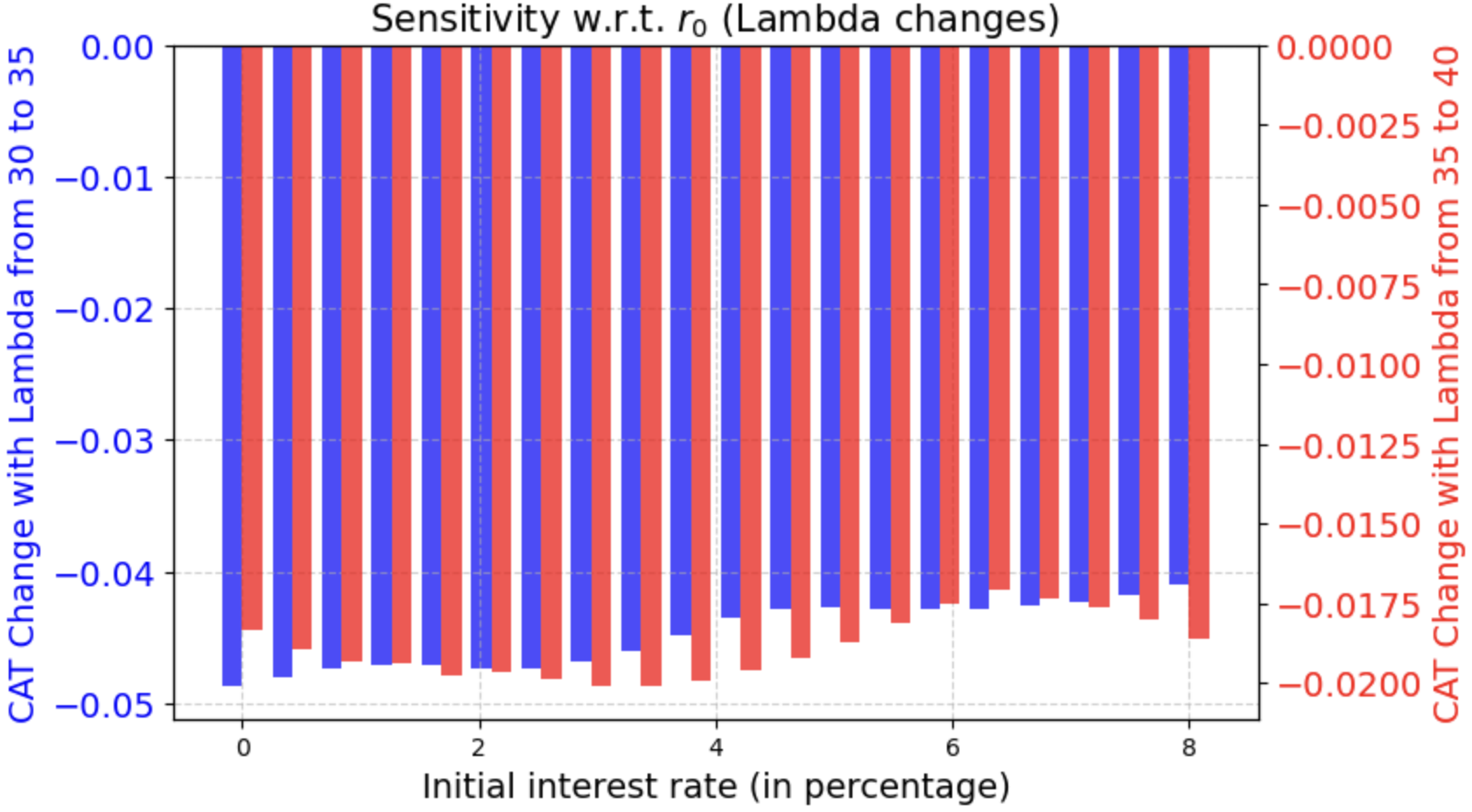}\caption{Rate changes}\end{subfigure}
\caption{Computed price sensitivities for Lognormal distributed losses. The same qualitative effects appear under the heavier-tailed severity specification.}
\label{fig:log_sens}
\end{figure}

\FloatBarrier
\section{Conclusion}
A neural surrogate turns structural CAT bond valuation into a real-time finance tool {for repeated in-domain valuation}. It preserves the interpretability of a compound-Poisson pricing model while amortizing the cost of rare-event simulation across future valuations. The empirical results show small out-of-sample pricing errors {on the simulated domain}, value estimates close to MC-IS benchmarks, and speed gains large enough to support screening and sensitivity analysis. The method is best viewed as a fast engine conditional on a chosen pricing measure, severity model, and trigger specification. Future work should combine the surrogate with market calibration, multi-peril dependence, richer trigger structures, and portfolio-level risk aggregation.

%\section*{Acknowledgments}
%The authors thank Prabhav Kumar for valuable input during the early stages of this work. Julian Sester gratefully acknowledges financial support by the NUS Start-Up Grant \emph{Tackling model uncertainty in Finance with machine learning}, by the  MOE AcRF Tier 1 Grant \emph{Neural Network-Based Valuation and Hedging of CAT Bonds under Model Uncertainty}, and by the MOE AcRF Tier 2 Grant \emph{Advances in Reinforcement Learning under Uncertainties: Toward Reliable Autonomous Decision-Making in Finance}

\section*{Acknowledgments}
The authors thank Prabhav Kumar for valuable input during the early stages of this work. J. Sester gratefully acknowledges support by the NUS Start-Up Grant \emph{Tackling model
uncertainty in Finance with machine learning}, by the MOE AcRF Tier 1 Grant 	 25-0428-P0001, and by the MOE AcRF Tier 2 Grant 	
T2EP20225-0030.

\section*{Data and code availability}
The reported results are based on simulated data generated from the model and parameter ranges described in the paper and supplementary material. No proprietary data were used. The replication code, training configurations, seeds, and notebooks are available at {\url{https://github.com/HuansangXu/CAT-bonds}}.

\section*{Declaration of generative AI and AI-assisted technologies in the manuscript preparation process}
During manuscript preparation, the authors used ChatGPT to assist with shortening and editorial rewriting. The authors reviewed and edited the content and take full responsibility for the submitted manuscript.

\appendix
\section*{Appendix: Supplementary material}
\addcontentsline{toc}{section}{Appendix: Supplementary material}

This supplement collects the technical statements, algorithms, proofs, and numerical diagnostics from the original manuscript. It is intended to be read alongside the main paper. In particular, Section~\ref{sec:nn} justifies the neural-network approximation by applying the universal approximation theorem to the CAT-bond pricing map on compact parameter sets. Sections~\ref{sec:MC}--\ref{sec:IS} explain the importance-sampling construction used to generate training labels, while the proofs of Propositions~\ref{prop_variance_reduction} and~\ref{prop_variance_reduction_logn} give the corresponding mathematical variance-reduction arguments for Gamma and Lognormal severity specifications. {Section~\ref{sec:numerical_diagnostics} reports additional diagnostics including boundary-region errors, MC-IS benchmark settings, offline computational costs, monotonicity checks, and architecture robustness.}

\section{Model and pricing statements}
\subsection{Modeling interest rates}\label{subsec:interest_rates}

Consider a probability space $(\Omega, \mathbb{G}, \mathbb{Q})$, where $\mathbb{Q}$ is interpreted as a risk-neutral probability measure, or spot martingale measure, and let $\mathbb{F} = \{\mathcal{F}(t)\}_{t \geq 0}$ be a sub-filtration of $\mathbb{G}$, describing the evolution of available market information over time. In line with \cite{cuchiero2010affine, Duffie1996AYM, duffee2002term, keller2011affine, piazzesi2010affine}, we use a general affine term structure model to model the spot interest rate, i.e., we assume the annual interest rate $r_t$ has dynamics 
\begin{align}
\dee r_t 
%&= \mu(t, r_t) \,\dee t + \sigma(t, r_t) \,\dee W_t \nonumber \\
&= (\alpha(t) r_t + \beta(t)) \,\dee t + \sqrt{\gamma(t) r_t + \delta(t)} \,\dee W_t,   \label{eq:affine_coefficients}
\end{align}
for $\{W_t\}_{t \geq 0}$ denoting a standard Brownian motion w.r.t.\,$\mathbb{F}$, 
 meaning that  the drift $\alpha(t) r_t + \beta(t)$ and variance rate $\gamma(t)  r_t + \delta(t)$ are affine in the spot rate $r_t$ at all times -- explaining the name of this model class.
% We assume that the annual interest rate $r_t$ satisfies the Vasicek model (\cite{Vasicek1977AE}), in which the diffusion process is a mean-reverting version of the Ornstein-Uhlenbeck process. The short-term interest rate $r$ is defined as the unique strong solution of the stochastic differential equation
% \begin{equation}
% \dee r_t = \alpha (\mu_r - r_t) \,\dee t + \sigma \,\dee W_t, \label{Vasicek}   
% \end{equation}
% where $\alpha$, $\mu_r$ and $\sigma$ are strictly positive constants. Under Vasicek's model, $\alpha$ is the speed of mean reversion, $\mu_r$ is the long-term mean level of interest rates, and $\sigma$ is the volatility of interest rate changes. 

% Hence, if the savings account $B(t)$ is continuously compounded at $r(t)$ from $t=t_0$ to $t_1$, then $B'(t) = r(t)B(t)$, i.e., 
% \begin{equation}
% B(t_1) = B(t_0)e^{\int_{t_0}^{t_1} r(s) \,\dee s}. \label{interestrateODE}
% \end{equation}
A non-defaultable zero coupon bond promises the payment of the face value $F$ at maturity $T$. Assuming no default risk, by risk-neutral valuation (see, e.g., \cite{brigo2006interest}), the price of the zero coupon bond at time $t \leq T$ is given by 
\begin{align}
P_{\rm Z}(F, t, T):= \mathbb{E}_{\mathbb{Q}}\left[Fe^{-\int_{t}^{T} r_s \,\dee s} \Big| \mathcal{F}(t) \right]
=Fe^{A(t,T)-B(t,T)r_t}, \label{zero coupon bond equation}
\end{align}
where the last equality follows from our assumption of an affine term structure, see also, e.g., \cite{keller2011affine}. $A(t,T)$ and $B(t,T)$ are deterministic functions of the model parameters which satisfy the ordinary system of equations:
\begin{equation} \label{ATSM_functions}
\begin{aligned}
1+ \frac{ \dee B(t,T)}{ \dee t} + \alpha(t) B(t,T) - \frac{1}{2} \gamma(t) B^2(t,T) &= 0, \qquad B(T,T) = 0  \\
\frac{ \dee A(t,T)}{\dee t} - \beta(t) B(t,T) + \frac{1}{2} \delta(t) B^2(t,T) &= 0, \qquad A(T,T)=0 
\end{aligned}  
\end{equation}
see, e.g., \cite{bjork2009arbitrage,keller2011affine}.
In order to ensure the existence of continuous solutions $A$ and $B$ we impose the following assumption.
\begin{asu}\label{asu:ATSM_continuous}
The coefficient functions $\alpha$, $\beta$, $\gamma$, $\delta$ from \eqref{eq:affine_coefficients} are assumed to be continuous.
\end{asu}

\begin{lem}\label{ATSM_continuous}
Under Assumption \ref{asu:ATSM_continuous}, there exist functions $A,B \in C(\mathbb{R}_+^2)$ that solve the Riccati equations~\eqref{ATSM_functions}. In particular, the solutions $A(t,T)$ and $B(t,T)$ are jointly continuous in $(t,T)$.
\end{lem}

The above described general types of affine term structure models include, in particular, the widely used Vasicek (\cite{Vasicek1977AE}) and Cox--Ingersoll--Ross (CIR)  (\cite{cox1985theory}) models which describe short-rate dynamics using a single stochastic factor driving the randomness.

\subsection{Pricing of CAT bonds}\label{subsec:pricing_cat}
Given the rare and unpredictable nature of catastrophic events, Poisson processes provide a natural mathematical framework for CAT bond pricing. In this approach, the jumps of the Poisson process correspond to catastrophe occurrences, with the jump frequency parameter capturing the rate at which catastrophes occur. This modeling choice has become the standard approach in the literature, as demonstrated in the comprehensive review by \cite{sukono2022application}, which examines various CAT bond pricing methodologies.

To describe these catastrophes mathematically, we model the number of catastrophes by a \textit{$\mathbb{F}$-conditional Poisson process}  $\{M(t)\}_{t \geq 0}$, defined on $(\Omega,\mathbb{G},\mathbb{Q})$, associated with positive intensity $\{\lambda(t)\}_{t \geq 0}$ being a $\mathbb{F}$-progressively measurable process with locally integrable sample path.

This allows to define a \textit{hazard process} $\{\Lambda(t): \Lambda(t) = \int_0^t \lambda(s) \, \D s \}_{t \geq 0}$, which is $\mathbb{F}$-adapted, right-continuous and increasing. For any $t>0$ and every $n=0,1,\cdots$, we then have
\begin{equation*}
\mathbb{Q}(M(t)=n | \mathcal{F}(t)) = \frac{(\Lambda(t))^n}{n!} e^{-\Lambda (t)}.    
\end{equation*}

Next, we model the indemnity loss trigger. The $i$th occurrence of a loss is given by a value drawn from the random variable $X_i$. For tractability, we assume that the losses \{$X_i\}$ are independent and identically distributed (i.i.d.) with cumulative distribution function (cdf) $F_{X}$ and independent of $\{M(t)\}_{t \geq 0}$. Then, the aggregate indemnity loss is a compound Poisson process $\{L(t)\} _{t \geq 0}$ such that 
\[
L(t) := \sum_{i=1}^{M(t)}{X_i}.
\]
A trigger event occurs exactly when $L(t)$ exceeds some pre-determined threshold $D > 0$, i.e., we are considering the stopping time $\tau = \inf\{t: L(t) \geq D\}$ indicating a loss event. We define $\{N(t)\} _{t \geq 0}$ by $N(t):= I(L(t) \geq D) = I(\tau \leq t)$, where $I$ denotes the indicator function $I(\tau \leq t) := 1$ if $\tau \leq t$, and $0$ otherwise, and we observe that $N$ is itself a doubly stochastic Poisson process, see also \cite{baryshnikov1998CAT}.

\begin{lem}\label{lem:double_stochastic_poisson}
The process $\{N(t)\} _{t \geq 0} = \{I(L(t) \geq D)\} _{t \geq 0}$ is a doubly stochastic Poisson process with intensity $\Bigg\{\lambda(t)\Bigl(1-F_X\bigr(D - L(t)\bigr)\Bigr)I\bigl(L(t) < D\bigr) \Biggr\}_{t \geq 0}$.
\end{lem}

This result leads directly to the following pricing formulas.

\begin{cor} \label{cor:cat_bond_pricing}
Consider a CAT-bond at time $t \leq T$  with face value $F$, maturity $T$, and loss threshold $D$.

\begin{enumerate}
    \item[(i)]
    The (risk-neutral) price of a zero-coupon CAT bond that pays $F$ at maturity only if no trigger event occurs by time $T$ is given by
    \begin{align}
    C_{\rm Z}(F, t, T)
    &= \left(1 - \mathbb{E}_{\mathbb{Q}}\left[N(T) \,\big|\, \mathcal{F}(t)\right]\right) \cdot P_{\rm Z}(F, t, T) \nonumber \\
    &= \mathbb{E}_{\mathbb{Q}}\left[1 - \int_{t}^{T} \lambda(s)\left(1 - F_X(D - L(s))\right) I(L(s) < D) \,\dee s \,\Big|\, \mathcal{F}(t) \right] \cdot P_{\rm Z}(F, t, T). \label{zero coupon catastrophe bond price}
    \end{align}
    \item[(ii)]
    The (risk-neutral) price of a CAT bond with face value $F$, maturity $T$, and coupon payments $\{C_i\}_{i=1}^n$ at times $\{t_i\}_{i=1}^n$ where $0 < t_1 < \cdots < t_n \leq T$, is given by
    \begin{equation}
    C_{\rm B}(F, \{C_i\}_{i=1}^n, t, T) = C_{\rm Z}(F, t, T) + \sum_{i=1}^n I(t \leq t_i) \cdot C_{\rm Z}(C_i, t, t_i). \label{coupon catastrophe bond equation}
    \end{equation}

    \item [(iii)]
    If, in the setting of (ii), the CAT bond has a recovery rate $R \in [0,1]$, then its (risk-neutral) price is
    \begin{align}
    C_{\rm RB}(F, \{C_i\}_{i=1}^n, t, T) 
    &= \sum_{i=1}^n I(t \leq t_i) \cdot C_{\rm Z}(C_i, t, t_i) 
    + (1 - \mathbb{E}_{\mathbb{Q}}[R]) \cdot C_{\rm Z}(F, t, T) 
    + \mathbb{E}_{\mathbb{Q}}[R] \cdot P_{\rm Z}(F, t, T). \label{recoverable coupon catastrophe bond equation}
    \end{align}
\end{enumerate}
\end{cor}

\section{Monte Carlo and importance sampling statements}\label{sec:MC}
\subsection{Monte Carlo integration}

According to \eqref{zero coupon catastrophe bond price}, prices of CAT bonds can be understood as a combination of zero coupon bond price $P_{\rm Z} (F, t, T)$ and the expectation $\mathbb{E}_{\mathbb{Q}}\left[N(T) | \mathcal{F}(t)\right]$ introducing path dependency, and thus imposing an additional numerical challenge. Due to Lemma~\ref{ATSM_continuous}, an explicit solution for the zero coupon bond price according to Eq.~\eqref{zero coupon bond equation} is readily available, however, in general there is no closed-form analytical solution for $\mathbb{E}_{\mathbb{Q}}\left[N(T) \,\big|\, \mathcal{F}(t)\right]$, so thus $C_{\rm Z} (F, t, T)$. Instead, a numerical approximation, e.g., via Monte Carlo integration is required to evaluate this expectation.
Alternatively, leveraging the Feynman--Kac theorem (see, e.g., \cite{cont2003financial}), integro-PDE methods could be applied (see, e.g., \cite{baryshnikov1998CAT}), which are used as a benchmark in the main manuscript.

To determine CAT bond prices at time $0$, due to \eqref{zero coupon catastrophe bond price}, we need to compute the risk-neutral probability of a trigger event that corresponds to $\theta := \mathbb{E}_{\mathbb{Q}} \left[N(t)\right] = \mathbb{E}_{\mathbb{Q}}\left[I (L(t) \geq D )\right] = \mathbb{Q}(L(t) > D)$. 

We approximate this probability using Monte Carlo simulation. For a fixed number of simulations $n \in \mathbb{N}$, we generate independent loss realizations $L_1(t),\dots,L_n(t)$ and compute  the estimator
$$
\hat{\theta}_n^{\rm{MC}} = \frac{1}{n} \sum_{i=1}^n I (L_i(t) \geq D ), 
$$
for $L_i(t)$ denoting the loss in the $i$th simulation. 
The variance of the estimated trigger probability is given as
\begin{align*}
\rm{Var}(\hat{\theta}_n^{\rm{MC}}) &= \frac{1}{n} \Big( \mathbb{E}_{\mathbb{Q}} \big[ I(L(t) \geq D)^2 \big] - \mathbb{E}_{\mathbb{Q}} \big[ I(L(t) \geq D) \big]^2  \Big) 
= \frac{1}{n} ( \theta - \theta^2 ).    
\end{align*}
%since $N(t)$ is an indicator variable and the variance follows from the Bernoulli variance formula.
Algorithm \ref{alg:MC_trigger} in the appendix outlines the procedure to compute this estimate.

\subsection{Importance sampling}\label{sec:IS}
By the law of large numbers, $\hat{\theta}^{\rm{MC}}$ converges to $\theta$ as the sample size increases. Although the MC estimator converges to $\theta$, plain Monte-Carlo estimation is characterized by slow convergence and a large variance, since for large thresholds $D$, very few samples contribute to the estimate, see also, e.g., \cite{glasserman2003MC}.
This motivates to apply \emph{Importance Sampling} (IS), a variance reduction technique  allowing to sample under a different probability measure making exceedance of the threshold $D$ a more frequently occurring event. To this end, we consider a second probability density $g$, whose support is assumed to contain that of original density $f$, and define the \emph{likelihood ratio} $R(l)$ by $R(l):=f(l)/g(l)$ whenever $g(l)>0$, and $R(l)=0$ otherwise. Hence, we get $\int I(\ell \geq D) f(\ell) \D \ell = \int I(\ell \geq D) g(\ell) R(\ell) \D \ell$.

The idea of using IS is to choose the density $g$ such that Monte-Carlo simulation under $g$ leads to an IS estimator $\hat{\theta}^{\rm{IS}}$ with variance being smaller than for the standard Monte Carlo integration. 
Then, we estimate $\hat{\theta}_n^{\rm{IS}} := \frac{1}{n} \sum_{i=1}^n I (L_i(t) \geq D) R \left(L_i(t) \right)$, where $L_i(t)$ is simulated according to density $g$, $I (L_i(t) \geq D)$ is the indicator function associated to the trigger event and $R \left(L_i(t) \right)$ denotes the likelihood ratio used for the change of measure in the $i$th simulation.
We compute that the variance for the estimated trigger probability equals
\begin{align*}
\rm{Var}(\hat{\theta}^{\rm{IS}}) &= \frac{1}{n} \Big( \mathbb{E}_g \big[ I(L(t) \geq D)^2 R \left(L(t) \right)^2 \big] - \mathbb{E}_g \big[ I(L(t) \geq D) R \left(L(t) \right) \big]^2  \Big) \\
&=\frac{1}{n} \Big( \mathbb{E}_g \big[ I(L(t) \geq D)^2 R \left(L(t) \right)^2 \big] - \theta^2  \Big)=\frac{1}{n} \Big( \mathbb{E}_\mathbb{Q} \big[ R \left(L(t) \right) I(L(t) \geq D) \big] - \theta^2  \Big).
\end{align*}
To determine a candidate change of measure for IS in the compound Poisson setting, we follow the procedure outlined in \cite[Chapter 7]{glasserman2003MC}, \cite[Chapter 11.4 ]{embrechts2011quantitative}, \cite{blanchet2007importance}, or \cite[Section 3]{LUTKEBOHMERT2025104598}, by applying \emph{exponential tilting}.
For the aggregated loss $L(t)$, two types of random variables need to be simulated - the number of losses and the size of losses. Firstly, consider the number of losses $M(t)$ till time $t$, modelled by a Poisson process with intensity $\{\lambda(t)\}_{t \geq 0}$.
The moment generating function (MGF)  of $M(t)$ is calculated to be 
\begin{equation*}
M_M(a) = \E_{\Q}\left[e^{aM(t)}\right]=e^{(\Lambda (t)) (e^{a}-1)},    \quad a \in \R,
\end{equation*}
where $\Lambda(t)$ is defined in Section \ref{subsec:pricing_cat}.
Applying exponential tilting that changes the probability measure from $\mathbb{Q}$ to $\mathbb{Q}_a$ for $a \in \mathbb{R}$, we have
\begin{align*}
\mathbb{Q}_a(M(t)=n | \mathcal{F}_t) &:= \mathbb{E}_{\mathbb{Q}} \Big[e^{a M(t)}/M_M(a) I({M(t)=n)} | \mathcal{F}_t \Big] \\
&= \frac{e^{an}}{e^{\Lambda (t) (e^{a}-1)}} \cdot \frac{(\Lambda (t))^n}{n!} e^{-\Lambda (t)} = \frac{(\Lambda(t) e^a)^n}{n!} e^{-\Lambda(t) e^a}.
\end{align*}

The number of losses $M(t)$ under probability measure $\mathbb{Q}_a$ is thus modelled by a Poisson process with new intensity $\{\lambda(t) e^a \}_{t \geq 0}$ after having applied exponential tilting. 

Next, we address the jump size component. Catastrophic losses typically exhibit heavy-tailed behavior, necessitating appropriate distributional choices. The actuarial literature has identified several continuous distributions suitable for modeling such extreme losses. Among these, the Gamma and Lognormal distributions have emerged as particularly popular choices, though other options include the Pareto and Burr distributions, compare, also \cite{Burnecki2003PricingOZ}. In the following subsections, we develop specific IS techniques for both the Gamma and Lognormal distributions.

\subsubsection{Importance sampling for Gamma distributed losses} \label{sec_IS_gamma}

We first consider Gamma distributed losses. Assume the size of losses $\{X_i\}_i$ are i.i.d.\, according to a Gamma distribution with shape $k$ and scale $\beta$, i.e., its probability density function is given as
\begin{equation*}
f_X(x) = \frac{1}{\Gamma(k) \beta^k} x^{k-1} e^{-x/\beta}, \quad x>0, k>0, \beta>0, \qquad \text{for }\Gamma(k):= \int_0^\infty t^{k-1} e^{-t} \D t.  
\end{equation*}
Applying exponential tilting, we have that the new density function for the size of losses under the new probability measure $\mathbb{Q}_{a,b}$ is given by
\begin{align*}
f_{a,b}(x) &:= \frac{\exp(bx)f_X(x)}{M_X(b)}
%&= \frac{\exp(bx)}{(1-\beta b)^{-k}} \frac{1}{\Gamma(k) \beta^k} x^{k-1} e^{-x/\beta} \\
= \frac{1}{\Gamma(k)} (\frac{1}{\beta}-b)^k x^{k-1} e^{-x(\frac{1}{\beta}-b)},
\end{align*}
for $b < \frac{1}{\beta}$, see, e.g., \cite{Asmussen2007StochasticSA}. This means the jump size satisfies a Gamma distribution with shape $k$ and new scale $\frac{\beta}{1-\beta b}$ under the probability measure $\mathbb{Q}_{a,b}$. 

Recall, that the initial mean value of the losses is usually much smaller than the predetermined threshold $D>0$, i.e., we have $\mathbb{E}_{\mathbb{Q}}[L] \ll D$, leading to a high variance of standard Monte Carlo estimates. For IS, we select parameters $a$ and $b$ such that the expected value of the loss under the modified measure $\mathbb{Q}_{a,b}$ approximates the threshold $D$, i.e., $\mathbb{E}_{\mathbb{Q}_{a,b}}[L] \approx D$. By this choice the event $\{L(t) > D\}$ will occur much more frequently which in turn helps to reduce the variance of the estimators. To this end, note that the expected loss under $\mathbb{Q}_{a,b}$ computes as
\begin{equation}
\mathbb{E}_{\mathbb{Q}_{a,b}}[L(t)] = \mathbb{E}_{\mathbb{Q}_{a,b}}[M(t)] \mathbb{E}_{\mathbb{Q}_{a,b}}[X_1] = (\lambda t e^a) \cdot k \frac{\beta}{1-\beta b}, \label{IS_gamma_loss} 
\end{equation}
where the first equality follows from Wald's identity.

From Eq.~\eqref{IS_gamma_loss}, it is evident that the choice of the importance sampling parameters $(a,b)$ plays a critical role in reducing the variance of the estimated trigger probability, which is the primary objective of importance sampling. While one may rely on numerical experimentation to evaluate various combinations of $(a,b)$ and select the most effective pair for a given loss distribution, this approach lacks generality and may not guarantee robustness across different scenarios. Therefore, it is important to investigate the regularity and consistency of this parameter selection procedure. Below, we present a result that guides the appropriate selection of $(a,b)$ across a broad class of loss distributions allowing for a more systematic and reliable guidance than case-by-case numerical tuning.

\begin{prop} \label{prop_variance_reduction}
Let an aggregate indemnity loss $\{L(t)\}_{t \geq 0}$ be defined by $L(t):= \sum_{i=1}^{M(t)} X_i$, where $M(t) \sim \rm{Poisson}(\lambda t)$ and $\{X_i\}_i$ are i.i.d.\ with  $X_1 {\sim} \rm{Gamma}(1,\beta)$ for $\beta \geq 0$. Further, let
\begin{equation*}
a= \frac{1}{2}\ln \left(\frac{D}{\lambda t \beta} \right),\qquad b = \frac{1}{\beta}-\frac{\lambda t e^a }{D}.
\end{equation*}
Then, the importance sampling estimator
$\hat{\theta}_n^{\rm{IS}} := \frac{1}{n} \sum_{i=1}^n I (L_i(t) \geq D) R \left(L_i(t) \right)$
with
$$
R(L(t))= \exp\left(\lambda t (e^a - 1) - a M(t)\right) (1-\beta b)^{- M(t)} \exp{ \big( -b L(t) \big)}
$$
satisfies
\begin{equation*}
\rm{Var}(\hat{\theta}^{\rm{IS}}) \leq \rm{Var}(\hat{\theta}^{\rm{MC}}), \qquad \E[\hat{\theta}^{\rm{IS}}] = \E[\hat{\theta}^{\rm{MC}}].
\end{equation*}
\end{prop}
IS is most beneficial when estimating probabilities of rare events, particularly when the expected loss under the original measure is significantly smaller than the threshold, i.e., when  $\mathbb{E}_{\mathbb{Q}}[L(t)] \ll D$. When the maturity $T$ is large, the aggregate indemnity loss $L(T)$ increases, leading to a higher trigger probability, making it not necessarily a rare event anymore. In such scenarios, one can use standard Monte Carlo simulation with a reasonable small variance instead.

In the appendix, we provide Algorithm~\ref{alg:IS_trigger_gamma} to estimate the trigger probability with Gamma distributed losses by using Monte Carlo simulation with IS.

\subsubsection{Importance sampling for Lognormally distributed losses} \label{sec_IS_lognormal}

Now we consider another frequently used type of loss distribution - the Lognormal distribution. Assume the size of losses $\{X_i\}_i$ are i.i.d. and satisfy a Lognormal distribution with parameters $\mu$ (mean of $\ln (X_i)$) and  $\sigma$ (standard deviation of $\ln(X_i)$). The probability density function of the random variable $X_i$ is
\begin{equation*}
f_X(x) = \frac{1}{x \sigma \sqrt{2\pi}} \exp\Big(- \frac{(\ln(x)-\mu)^2} {2 \sigma^2}\Big) \quad x>0, \sigma>0, \mu \in \mathbb{R}.    
\end{equation*}
The moment generating function of a Lognormal distribution is infinite for positive arguments. Thus, we cannot apply exponential tilting for Lognormal distributions directly. Instead, we follow the idea from \cite{Bee2009IS}, where the importance sampling density is built as another Lognormal distribution with a larger expected value, denote as $X \sim \text{Logn}(\mu+b, \sigma^2)$ under the probability measure $\mathbb{Q}_{a,b}$, with $ b \in \mathbb{R}^+$. 
% The likelihood ratio in this situation is given as
% \begin{equation*}
% \frac{\frac{1}{x \sigma \sqrt{2\pi}} 
% \exp\Big(- \frac{(\ln(x)-\mu)^2}{2 \sigma^2}\Big)}
% {\frac{1}{x \sigma \sqrt{2\pi}} \exp\Big(- \frac{(\ln(x)-(\mu+b))^2}
% {2 \sigma^2}\Big)}
% = \exp\Big(\frac{b^2 + 2\mu b - 2 b \ln(x)}{2 \sigma^2}\Big).
% \end{equation*}
Analogue to Eq.~\eqref{IS_gamma_loss}, we compute the expected loss under $\mathbb{Q}_{a,b}$, and obtain 
\begin{equation}
\mathbb{E}_{\mathbb{Q}_{a,b}}[L(t)] = \mathbb{E}_{\mathbb{Q}_{a,b}}[M(t)] \mathbb{E}_{\mathbb{Q}_{a,b}}[X_1] = (\lambda t e^a) \cdot \exp\big(\mu+b+\frac{\sigma^2}{2}\big). \label{IS_lognormal_loss}      
\end{equation}

Due to the complexity and heavy-tailed nature of the Lognormal distribution, it is not feasible to directly propose suitable IS parameters for Lognormal losses. While in the Gamma case explicit IS parameters can be derived through Proposition \ref{prop_variance_reduction}, the Lognormal setting requires a different strategy. Therefore, we adopt a variance minimization approach to determine the optimal IS parameters for Lognormal distributed losses. Further, recall that the variance of IS estimator is given as
\begin{align*}
\rm{Var}(\hat{\theta}^{\rm{IS}}) =\frac{1}{n} \Big( \mathbb{E}_\mathbb{Q} \big[ R(L(t)) I(L(t) \geq D) \big] - \theta^2  \Big),
\end{align*}
and hence our goal is minimize $\mathbb{E}_\mathbb{Q} \big[ R(L(t)) I(L(t) \geq D) \big]$. Under $\mathbb{Q}_{a,b}$ we have $M(t) \sim \rm{Poisson}(\lambda te^a)$ and $X_i \overset{\text{i.i.d.}}{\sim} {\rm Logn}(\mu+b,\sigma^2)$. Thus, we can compute the likelihood ratio under the measure $\mathbb{Q}_{a,b}$ after IS as follows:
\begin{equation}\label{eq:R_decom}
R(L(t)) = \exp \Big(\lambda t (e^a - 1) - a M(t) \Big)\exp \Big(\frac{b^2} {2\sigma^2} M(t)\Big) \exp \Big( - \frac{b} {\sigma^2} \sum_{i=1}^{M(t)} \big(\ln (X_i) -\mu \big)\Big).    
\end{equation}
By minimizing the expectation $\mathbb{E}_\mathbb{Q} \big[ R(L(t)) I(L(t) \geq D) \big]$, we get the following result

\begin{prop} \label{prop_variance_reduction_logn}
Let an aggregate indemnity loss $\{L(t)\}_{t \geq 0}$ be defined by $L(t):= \sum_{i=1}^{M(t)} X_i$, where $M(t) \sim \rm{Poisson}(\lambda t)$ and $X_i \overset{\text{i.i.d.}}{\sim} \rm{Logn}(\mu, \sigma^2)$ with $\mu \in \R, \sigma \geq 0$. Set $a, b \geq 0$ such that
\begin{equation*}
\frac{2 \lambda t \, b}{\sigma^2} \, e^{\frac{b^2}{2\sigma^2}} - \frac{1}{\sigma} \, \frac{\phi\Big(\frac{\ln D - \mu + b}{\sigma}\Big)}{1 - \Phi\Big(\frac{\ln D - \mu + b}{\sigma}\Big)} = 0 \quad \text{and} \quad a= \frac{b^2}{2\sigma^2}, 
\end{equation*}
where $\phi(\cdot)$ and $\Phi(\cdot)$ are the pdf and cdf of the standard normal distribution, respectively.
Then for $D$ large enough, the variance of the importance sampling estimator 
$\hat{\theta}_n^{\rm{IS}} := \frac{1}{n} \sum_{i=1}^n I (L_i(t) \geq D) R \left(L_i(t) \right)$
with $R(L(t))$ as in \eqref{eq:R_decom} satisfies
\begin{equation*}
\rm{Var}(\hat{\theta}^{\rm{IS}}) <\rm{Var}(\hat{\theta}^{\rm{MC}}), \qquad \E[\hat{\theta}^{\rm{IS}}] = \E[\hat{\theta}^{\rm{MC}}].
\end{equation*}
\end{prop}

The optimal $b$ pushes the log-severity mean close to $\ln D$, $a$ then follows from 
$b$ to keep the Poisson weight balanced.
Algorithm~\ref{alg:IS_trigger_lognormal} for estimating the trigger probability with Lognormally distributed losses by using Monte Carlo simulation with IS is provided in the appendix.

Having established the Monte Carlo simulation algorithms, we can utilize the estimated trigger probability $\mathbb{E}_{\mathbb{Q}}\left[N(T)\right]$ in conjunction with the zero-coupon bond price from Eq.~\eqref{zero coupon bond equation} under an affine term structure interest rate model. By incorporating these components into the CAT bond pricing formula in Eq.~\eqref{coupon catastrophe bond equation}, we can determine the CAT bond price under various parameter settings. The complete calculation procedure is outlined in Algorithm \ref{alg:CAT_price}.

\section{Pricing CAT bonds with neural networks}\label{sec:nn}

\subsection{Feedforward neural networks}
The following exposition on neural networks~(NNs) is borrowed from \cite{neufeld2023neural} but can be found in a similar form in any standard textbook on the topic, e.g., \cite{goodfellow2016deep} or \cite{nielsen2015neural}.
By fully-connected feedforward neural networks, or simply \textit{neural networks} with input dimension $d_{\operatorname{in}} \in \N$, output dimension $d_{\operatorname{out}} \in \N$, and number of layers $l \in \N$ we refer to functions of the form
\begin{equation}\label{eq_nn_function}
	\begin{aligned}
		\R^{d_{\operatorname{in}}} &\rightarrow \R^{d_{\operatorname{out}}}\\
		{x} &\mapsto {A_l} \circ {\varphi}_l \circ {A_{l-1}} \circ \cdots \circ {\varphi}_1 \circ {A_0}({x}),
	\end{aligned}
\end{equation}
where $({A_i})_{i=0,\dots,l}$ are affine\footnote{This means for all $i=0,\dots,l$, the function ${A_i}$ is assumed to have an affine structure of the form
	$
	{A_i}({x})={M_i} {x} + {b_i}
	$
	for some matrix ${M_i} \in \R^{ h_{i+1} \times h_{i}}$ and some vector ${b_i}\in \R^{h_{i+1}}$, where $h_0:=d_{\operatorname{in}}$ and $h_{l+1}:=d_{\operatorname{out}}$. } functions of the form 
\begin{equation}\label{eq_A_i_def}
	{A_0}: \R^{d_{{\operatorname{in}}}} \rightarrow \R^{h_1},\qquad {A_i}:\R^{h_i}\rightarrow \R^{h_{i+1}}\text{ for } i =1,\dots,l-1, \text{(if } l>1), \text{ and}\qquad {A_l} : \R^{h_l} \rightarrow \R^{d_{\operatorname{out}}},
\end{equation}
and where the function $\varphi_i$ is applied componentwise, i.e., for  $i=1,\dots,l$ we have ${\varphi}_i(x_1,\dots,x_{h_i})=\left(\varphi(x_1),\dots,\varphi(x_{h_i})\right)$.  The function $\varphi:\R \rightarrow \R$  is called \emph{activation function} and assumed to be continuous and non-polynomial.
%We say a neural network is \emph{deep} if $l\geq 2$.
Here ${h}=(h_1,\dots,h_{l}) \in \N^{l}$ denotes the dimensions (the number of neurons) of the hidden layers, also called \emph{hidden dimension}.

Then, we denote by $\mathfrak{N}_{d_{\operatorname{in}},{d_{\operatorname{out}}}}^{l,{h}}$  the set of all neural networks with input dimension ${d_{\operatorname{in}}}$, output dimension ${d_{\operatorname{out}}}$, $l$ hidden layers, and hidden dimension ${h}$, whereas
the set of all neural networks from $\R^{d_{\operatorname{in}}}$ to $\R^{d_{\operatorname{out}}}$ (i.e.\ without specifying the number of hidden layers and hidden dimension) is denoted by
\[
\mathfrak{N}_{d_{\operatorname{in}},{d_{\operatorname{out}}}}:=\bigcup_{l \in \N}\bigcup_{{h} \in \N^l}\mathfrak{N}_{d_{\operatorname{in}},{d_{\operatorname{out}}}}^{l,{h}}.
\]

It is well-known that the set of neural networks possesses the so-called \textit{universal approximation property}, see, e.g., \cite{hornik1989multilayer} or \cite{pinkus1999approximation}.
\begin{prop}[Universal approximation theorem]\label{lem_universal}
For any compact set $\K \subset \R^{d_{\operatorname{in}}} $ the set $\mathfrak{N}_{d_{\operatorname{in}},{d_{\operatorname{out}}}}|_{\K}$ is dense in ${C}(\K,\R^{d_{\operatorname{out}}})$ with respect to the topology of uniform convergence on $C(\K,\R^{d_{\operatorname{out}}})$.
\end{prop}

By using Proposition~\ref{lem_universal} we can show that neural networks are able to approximate CAT bond prices arbitrarily well after being trained. Firstly, recalling the pricing formula for the zero coupon CAT bond Eq.~\eqref{zero coupon catastrophe bond price}, we assume that the Poisson intensity is fixed over time that $\lambda(t) \equiv \lambda$ and then we set $C_{\rm Z}(F,T;r_0, \lambda, D):=C_{\rm Z}(F, 0, T)$ to emphasize its dependence on the initial interest rate $r_0$, the Poisson intensity $\lambda$ and the trigger threshold $D$.

\begin{prop} \label{UA_proof}
Let Assumption~\ref{asu:ATSM_continuous} hold true, let $\mathbb{K}\subset \R^5$ be a compact set and let the loss distribution of $X_1$ follow a continuous distribution. Then for all $\varepsilon>0$, there exists a fully connected feedforward neural network $\mathcal{N}\mathcal{N}\in  \mathfrak{N}_{5,1}$ such that for all $k=(r_0,\lambda,T,F,D) \in \mathbb{K}$
\[
|\mathcal{N}\mathcal{N}(k)-C_{\rm Z}(F,T;r_0, \lambda, D)|< \varepsilon.
\]
\end{prop}

According to Eq.~\eqref{coupon catastrophe bond equation}, the coupon CAT bond price $C_{\rm B}(F, \{C_i\}_{i=1}^n, 0, T)$ is the cumulative sum of zero coupon CAT bonds with face value $F$ and $C_i$s, each at the respective maturity. Hence, we directly get the following proposition showing that the coupon CAT bond price $C_{\rm B}(F, \{C_i\}_{i=1}^n, 0, T)$ in Eq.~\eqref{coupon catastrophe bond equation} can be approximated arbitrarily well by neural networks whenever we restrict the inputs to a compact set. In order to emphasize the dependence of the coupon CAT bond price on the payment schedule $\{t_i\}_{i=1}^n$, the initial interest rate $r_0$, the Poisson intensity $\lambda$ and the threshold $D$, we set $C_{\rm B}(F, \{C_i\}_{i=1}^n, \{t_i\}_{i=1}^n, T; r_0, \lambda, D):=C_{\rm B}(F, \{C_i\}_{i=1}^n, 0, T)$

\begin{prop} \label{UA_proof_coupon}
Let Assumption~\ref{asu:ATSM_continuous} hold true.
Given $n \in \N$, let $\mathbb{K} \subset \R^{2n+5}$ be a compact set and let the loss distribution of $X_1$ follow a continuous distribution. Then for all $\varepsilon>0$, there exists a fully connected feedforward neural network $\mathcal{N}\mathcal{N}_B\in  \mathfrak{N}_{2n+5,1}$ such that for all $k_B=(r_0, \lambda,\{t_i\}_{i=1}^n,T,\{C_i\}_{i=1}^n,F,D) \in \mathbb{K}$
\[
\bigl|\mathcal{N}\mathcal{N}_B(k_B)-C_{\rm B}(F, \{C_i\}_{i=1}^n, \{t_i\}_{i=1}^n, T; r_0, \lambda, D) \bigr|< \varepsilon.
\]
\end{prop}

\begingroup
\section{Additional numerical diagnostics and reproducibility details}\label{sec:numerical_diagnostics}
\subsection{Operational domain, boundary regions, and extrapolation}
The training design in the main manuscript defines the operational training domain
\[
\mathcal D_{\rm train}=[0,0.08]\times[30,40]\times[7,13]\times[90,720]\times\{0,2,3,4,6,8,10,12\},
\]
where the coordinates are $(r_0,\lambda,D,T,N)$ and $T$ is measured in days. The neural-network surrogate is therefore assessed as an interpolator on this compact set. It is not used as a validated extrapolation rule outside \(\mathcal D_{\rm train}\). For out-of-domain contracts, the recommended procedure is to enlarge the simulation design and retrain, or to fall back to direct MC-IS valuation.

For each continuous coordinate $x_j \in [\ell_j,u_j]$, where $\ell_j$ and $u_j$ are the lower and upper bounds of the training domain, respectively, we define the boundary subset as the set of points whose $j$-th coordinate lies within the outermost $5\%$ of the interval near either boundary:
\[
\mathcal E_j
=
\left\{
x:
x_j \in
[\ell_j,\ell_j+0.05(u_j-\ell_j)]
\cup
[u_j-0.05(u_j-\ell_j),u_j]
\right\}.
\]
For the coupon-count coordinate, the boundary subset is $N\in\{0,12\}$. Table~\ref{tab:supp_boundary} reports pricing errors on these held-out boundary subsets demonstrating only a slight deterioration of accuracy near the boundaries.

\begin{table}[H]
\centering
\small
{
\caption{Boundary-region pricing errors on the held-out test sample}
\label{tab:supp_boundary}
\begin{tabular}{llrrrr}
\toprule
Severity & Test subset & Observations & MAE & RMSE & Max. AE \\
\midrule
Gamma & Full test set & 120000 & 0.00274 & 0.00380 & 0.02556 \\
Gamma & $r_0$ boundary & 11877 & 0.00304 & 0.00408 & 0.02108 \\ 
Gamma & $\lambda$ boundary & 12007 & 0.00291 & 0.00404 & 0.02348 \\
Gamma & $D$ boundary & 11864 & 0.00305 & 0.00413 & 0.02391 \\
Gamma & $T$ boundary & 12063 & 0.00328 & 0.00447 & 0.02324 \\
Gamma & $N\in\{0,12\}$ & 30080 & 0.00296 & 0.00407 & 0.02556 \\
\midrule
Lognormal & Full test set & 120000 & 0.00308 & 0.00429 & 0.03296 \\
Lognormal & $r_0$ boundary & 12081 & 0.00345 & 0.00464 & 0.03296 \\
Lognormal & $\lambda$ boundary & 12265 & 0.00328 & 0.00463 & 0.03296 \\
Lognormal & $D$ boundary & 11822 & 0.00328 & 0.00456 & 0.03296 \\
Lognormal & $T$ boundary & 11880 & 0.00366 & 0.00499 & 0.02918 \\
Lognormal & $N\in\{0,12\}$ & 29941 & 0.00333 & 0.00470 & 0.03296 \\
\bottomrule
\end{tabular}
}
\end{table}

\subsection{Expanded out-of-sample accuracy diagnostics}
Table~\ref{tab:supp_accuracy} complements Figure~1 in the main manuscript. Besides MAE and MSE, it reports signed bias, RMSE, high absolute-error quantiles, maximum absolute error, and $R^2$ on the held-out test set.

\begin{table}[H]
\centering
\small
{
\caption{Expanded out-of-sample error statistics}
\label{tab:supp_accuracy}
\begin{tabular}{lrrrrrrrr}
\toprule
Severity & Observations & Bias & MAE & MSE & RMSE & 95\% AE & 99\% AE & $R^2$ \\
\midrule
Gamma & 120000 & 0.00026 & 0.00274 & $1.4\times10^{-5}$ & 0.00380 & 0.00808 & 0.01286 & 0.99989 \\
Lognormal & 120000 & 0.00009 & 0.00308 & $1.8\times10^{-5}$ & 0.00429 & 0.00928 & 0.01413 & 0.99984 \\
\bottomrule
\end{tabular}}
\end{table}

\subsection{MC-IS benchmark protocol and offline computational cost}
For the timing benchmark in the main manuscript, MC-IS is run under a fixed and reproducible protocol. Each protected coupon date and the principal date are priced by estimating the corresponding trigger probability with the same importance-sampling construction as in Section~\ref{sec:IS}. The comparison is therefore between a direct MC-IS valuation of each requested contract and evaluation of the already-trained surrogate.

\begin{table}[H]
\centering
\small
\caption{MC-IS benchmark protocol}
\label{tab:supp_mcis_protocol}
{
\begin{tabular}{ll}
\toprule
Item & Specification \\
\midrule
Paths per trigger-probability estimate & {5,000 importance-sampling paths} \\
Stopping/convergence rule & { Fixed number of Monte Carlo simulations} \\
Random seeds & {NumPy random seed = 125} \\
Common random numbers across methods & {No. The surrogate evaluation is deterministic.} \\
Hardware for timing experiments & {Apple M3 Pro} \\
Software stack & {Python 3.12.2, NumPy 1.26.4, SciPy 1.13.1, TensorFlow 2.18.0} \\
\bottomrule
\end{tabular}}
\end{table}

The neural-network timings in the main manuscript are online evaluation times. Table~\ref{tab:supp_offline_cost} records the corresponding one-time offline costs. These costs are incurred once for a fixed model specification, severity law, and input domain, and are amortized when the trained pricing surface is queried repeatedly.

\begin{table}[H]
\centering
\small
{
\caption{Offline computational cost of label generation and training}
\label{tab:supp_offline_cost}
\begin{tabular}{llr}
\toprule
Task & Specification & Wall-clock time \\
\midrule
Label generation & Gamma severities, 600,000 prices & 24 h 21 min \\
Label generation & Lognormal severities, 600,000 prices & 33 h 20 min \\
Training & Architecture $(128,64,32)$, Gamma $+$ Lognormal & 21.55 min \\
Training & Architecture $(256,128,64,32)$, Gamma & 29.70 min \\
Training & Architecture $(256,128,64,32)$, Lognormal & 35.98 min \\
Training & Architecture $(512,256,128,64)$, Gamma $+$ Lognormal & 81.62 min \\
Full cross-validation & All candidate architectures and folds & 7 h 27 min \\
\bottomrule
\end{tabular}

\vspace{1mm}
\footnotesize
\emph{Note.} Full cross-validation evaluates four candidate configurations for each severity distribution (Gamma and Lognormal), obtained by combining two hidden-layer architectures, $(128,64,32)$ and $(256,128,64,32)$, with two learning rates, $10^{-4}$ and $10^{-5}$. The activation function (ReLU), L2 regularization coefficient ($10^{-4}$), dropout rate ($0.1$), and batch normalization are kept fixed throughout.
}
\end{table}

\subsection{Economic monotonicity diagnostics}
The structural model implies economically expected signs on the training domain: prices should be non-increasing in the catastrophe intensity $\lambda$, non-decreasing in the attachment threshold $D$, and non-increasing in the initial short rate $r_0$ for positive protected cash flows. We quantify whether the learned surrogate violates these signs on a dense finite-difference grid. With tolerance $\varepsilon_{\rm tol}=1\times10^{-6}$, a violation is recorded if
\[
\Delta_\lambda V>\varepsilon_{\rm tol},\qquad \Delta_D V<-\varepsilon_{\rm tol},\qquad \Delta_{r_0} V>\varepsilon_{\rm tol},
\]
where each difference compares adjacent grid values while holding the other inputs fixed.

{ In the following Table~\ref{tab:supp_monotonicity}, we evaluate the model on dense uniform grids consisting of 100,001 points for each parameter: $r_0 \in [0,0.12]$, $\lambda \in [25,45]$, and $D \in [5\times 10^9,15\times 10^9]$, which include samples out of the training domain. The table shows that no violations of the monotonicity relations can be observed.}

\begin{table}[H]
\centering
\small
{
\caption{Finite-difference monotonicity violations on dense grids}
\label{tab:supp_monotonicity}
\begin{tabular}{llrrrr}
\toprule
Severity & Variable & Expected sign & Grid comparisons & Violation rate & Max. violation \\
\midrule
Gamma & $\lambda$ & non-increasing & 100000 & 0\% & 0 \\
Gamma & $D$ & non-decreasing & 100000 & 0\% & 0 \\
Gamma & $r_0$ & non-increasing & 100000 & 0\% & 0 \\
Lognormal & $\lambda$ & non-increasing & 100000 & 0\% & 0 \\
Lognormal & $D$ & non-decreasing & 100000 & 0\% & 0 \\
Lognormal & $r_0$ & non-increasing & 100000 & 0\% & 0 \\
\bottomrule
\end{tabular}
}
\end{table}

\subsection{Stability to architecture selection}
Table~\ref{tab:supp_architecture} reports the performance of alternative architectures considered during model selection. The goal is not to claim that the selected architecture is unique, but to show whether valuation accuracy is stable across nearby depth/width choices.

% \begin{table}[H]
% \centering
% \small
% \caption{Architecture robustness on the held-out test set}
% \label{tab:supp_architecture}
% \begin{tabular}{llrrrr}
% \toprule
% Severity & Architecture & Parameters & MAE & MSE & Training time \\
% \midrule
% Gamma & $(128,64,32)$ & {\color{red}[p]} & 0.00450 & $3.6\times10^{-5}$ & 705.05 s \\
% Gamma & $(256,128,64,32)$ & {\color{red}[p]} & 0.00282 & $1.6\times10^{-5}$ & 2628.52 s \\
% Gamma & $(512,256,128,64)$ & {\color{red}[p]} & 0.00288 & $1.7\times10^{-5}$ & 3466.23 s \\
% Lognormal & $(128,64,32)$ & {\color{red}[p]} & 0.00383 & $2.5\times10^{-5}$ & 623.75 s \\
% Lognormal & $(256,128,64,32)$ & {\color{red}[p]} & 0.00321 & $1.9\times10^{-5}$ & 1475.54 s \\
% Lognormal & $(512,256,128,64)$ & {\color{red}[p]} & 0.00309 & $1.9\times10^{-5}$ & 3301.85 s \\
% \bottomrule
% \end{tabular}
% \end{table}

\begin{table}[H]
\centering
\small
{
\caption{Architecture robustness on the held-out test set}
\label{tab:supp_architecture}
\begin{tabular}{llrrrr}
\toprule
Severity & Architecture & Parameters & MAE & MSE & Training time \\
\midrule
Gamma & $(128,64,32)$ & 11,137 & 0.00331 & $2.5\times10^{-5}$ & 653.69 s \\
Gamma & $(256,128,64,32)$ & 44,801 & 0.00274 & $1.4\times10^{-5}$ & 1782.01 s \\
Gamma & $(512,256,128,64)$ & 175,617 & 0.00317 & $1.9\times10^{-5}$ & 2286.54 s \\
Lognormal & $(128,64,32)$ & 11,137 & 0.00379 & $2.6\times10^{-5}$ & 639.37 s \\
Lognormal & $(256,128,64,32)$ & 44,801 & 0.00308 & $1.8\times10^{-5}$ & 2158.52 s \\
Lognormal & $(512,256,128,64)$ & 175,617 & 0.00309 & $1.9\times10^{-5}$ & 2610.53 s \\
\bottomrule
\end{tabular}

\vspace{1mm}
\footnotesize
\emph{Note.} The parameter counts are trainable parameters reported by the final Keras model summary. All models use identical hyperparameters: ReLU activation, L2 regularization coefficient $10^{-4}$, dropout rate $0.1$, batch normalization, and learning rate $10^{-5}$. Only the hidden-layer architecture varies.
}
\end{table}

\section{Algorithms}

In this section, we present all the algorithms used. Implementations of the algorithms can be found for the convenience of the reader, and for reproducibility of the results, under {\url{https://github.com/HuansangXu/CAT-bonds}}.

\begin{algorithm}
\caption{Standard Monte Carlo Estimation of Trigger Probability $\mathbb{E}_{\mathbb{Q}}\left[N(T)\right]$}
\label{alg:MC_trigger}
\SetAlgoLined
\KwIn{Number of simulations $n_s$, Poisson intensity parameter $\lambda$, trigger threshold $D$, maturity $T$, Loss distribution.}
\KwOut{Estimated trigger probability $\hat{\theta}^{\rm{MC}}_{n_s}$.}
\vspace{0.2cm}

Generate $n_s$ Poisson-distributed random variables $N_i \sim \text{Poisson}(\lambda T)$;\\
Initialize empty list $h$;\\

\For{$i = 1$ \KwTo $n_s$}{
    Generate $N_i$ i.i.d. random variables $X_1, X_2, \cdots, X_{N_i}$ from the distribution of losses; \\
    Compute their sum $L_i = \sum_{i=1}^{N_i} X_i $; \\
    \eIf{$L_i \geq D$}{
        $h_i \leftarrow 1$ (trigger event).
    }{
        $h_i \leftarrow 0$ (non trigger).
    }
}

Compute cumulative average $\hat{\theta}^{\rm{MC}}_{n_s} = \frac{1}{n_s}\sum_{j=1}^{n_s} h_j$; \\

\Return $\hat{\theta}^{\rm{MC}}_{n_s}$.

\end{algorithm}

\begin{algorithm}[H]
\label{alg:IS_trigger_gamma}
\caption{IS Estimation of Trigger Probability $\mathbb{E}_{\mathbb{Q}}\left[N(T)\right]$ with Gamma Distributed Losses}
\KwIn{Number of simulations $n_s$, Poisson intensity $\lambda$, trigger threshold $D$, time horizon $T$, parameters for Gamma distribution shape $k$ and scale $\beta$}
\KwOut{Estimated trigger probability 
$\hat{\theta}_{n_s}^{\rm{IS}}$ using Monte Carlo simulation with importance sampling}

Compute $a_{\text{poi}} =\tfrac{1}{2} \ln\left(\tfrac{D}{ \lambda T  k \beta}\right)$ \;
Compute the new Poisson intensity $\lambda'$ after IS: $\lambda' = \lambda T e^{a_{\text{poi}}}$ \;
Compute $b_{\text{gam}} = \tfrac{1}{\beta} -\tfrac{\lambda T e^{a_{\text{poi}}}k}{D}$ \;
Compute the new scale $\beta'$ of Gamma distribution after IS $\beta' = 1/(1/\beta - b_{\text{gam}})$ \;

Initialize arrays $h$, $R^{\text{poi}}$, and $R^{\text{gam}}$ of size $n_s$ \;

\For{$i = 1$ \KwTo $n_s$}{
    Sample $N_i \sim \text{Poisson}(\lambda')$ \;
    Generate $L_i = \sum_{j=1}^{N_i} X_j$, where $X_j \sim \text{Gamma}(k, \beta')$ \;
    
    Compute $h_i = I(L_i \geq D)$ \;
    Compute the likelihood ratio for Poisson process $R^{\text{poi}}_i = \exp(\lambda T (e^{a_{\text{poi}}} - 1) - a_{\text{poi}} N_i)$ \;
    Compute the likelihood ratio for Gamma distribution $R^{\text{gam}}_i = (1 - \beta b_{\text{gam}})^{-k N_i} \exp( -b_{\text{gam}} L_i)$ \;
}

Compute cumulative average $\hat{\theta}_{n_s}^{\rm{IS}} = \frac{1}{n_s} \sum_{j=1}^{n_s} h_j \cdot R^{\text{poi}}_j \cdot R^{\text{gam}}_j$ \;

\Return{$\hat{\theta}_{n_s}^{\rm{IS}}$}
\end{algorithm}

\begin{algorithm}[H]
\label{alg:IS_trigger_lognormal}
\caption{IS Estimation of Trigger Probability $\mathbb{E}_{\mathbb{Q}}\left[N(T)\right]$ with Lognormally Distributed Losses}
\KwIn{Number of simulations $n_s$, Poisson intensity $\lambda$, trigger threshold $D$, time horizon $T$, parameters for Lognormal distribution location $\mu$ and scale $\sigma$}
\KwOut{Estimated trigger probability 
$\hat{\theta}_{n_s}^{\rm{IS}}$ using Monte Carlo simulation with importance sampling}

Compute $b_{\text{log}}$ by solving equation 
$\frac{2 \lambda t \, b}{\sigma^2} \, e^{\frac{b^2}{2\sigma^2}} - \frac{1}{\sigma} \, \frac{\phi\Big(\frac{\ln D - \mu + b}{\sigma}\Big)}{1 - \Phi\Big(\frac{\ln D - \mu + b}{\sigma}\Big)} = 0$ \;
Compute the new location $\mu'$ of Lognormal distribution after IS $\mu' = \mu + b_{\text{log}}$ \;
Compute $a_{\text{poi}} = b_{\text{log}}^2 / (2\sigma^2)$ \;
Compute the new Poisson intensity $\lambda'$ after IS: $\lambda' = \lambda T e^{a_{\text{poi}}}$ \;

Initialize arrays $h$, $R^{\text{poi}}$, and $R^{\text{log}}$ of size $n_s$ \;

\For{$i = 1$ \KwTo $n_s$}{
    Sample $N_i \sim \text{Poisson}(\lambda')$ \;
    Generate $L_i = \sum_{j=1}^{N_i} \exp{(Y_j)}$, where $Y_j \sim \mathcal{N}(\mu', \sigma^2)$ \;
    
    Compute $h_i = I(L_i \geq D)$ \;
    Compute the likelihood ratio for Poisson process $R^{\text{poi}}_i = \exp(\lambda T (e^{a_{\text{poi}}} - 1) - a_{\text{poi}} N_i)$ \;
    Compute the likelihood ratio for Lognormal distribution $R^{\text{log}}_i = \exp \Big(\frac{(\mu + b_{\text{log}})^2 - \mu^2} {2\sigma^2} \Big)^{N_i} \exp \Big(- \frac{b_{\text{log}} \sum_{j=1}^{N_i} Y_j} {\sigma^2} \Big)$ \;
}

Compute cumulative average $\hat{\theta}_i^{\rm{IS}} = \frac{1}{n_s}\sum_{j=1}^{n_s} h_j \cdot R^{\text{poi}}_j \cdot R^{\text{log}}_j$ \;

\Return{$\hat{\theta}_{n_s}^{\rm{IS}}$}
\end{algorithm}

\begin{algorithm}[H]
\label{alg:CAT_price}
\caption{CAT Bond Pricing using Monte Carlo Simulation for Trigger Probability}
\SetAlgoLined
\KwIn{Face value $F$, interest rate model parameters (initial interest rate $r_0$, drift term $\mu$, and volatility $\sigma$), Poisson intensity $\lambda$, trigger threshold $D$, loss distributions parameters ($\text{Gamma}(k,\beta)$ or $\text{Logn}(\mu, \sigma)$), number of simulations $n_s$ for trigger probability, coupon values $\{C_i\}_{i=1}^n$, coupon payment dates $\{t_i\}_{i=1}^n$, maturity $T$}
\KwOut{Estimated CAT Bond Price at time $0$: $C_{\rm B}$}

Initialize CAT bond price: $C_{\rm B} = 0$. \\

\For{$i = 1$ \KwTo $n$}{
    \If{$\mathbb{E}_{\mathbb{Q}}[L(t_i)] < D$}{
        Compute trigger probability $\theta_i$ at payment date $t_i$ using Monte Carlo simulation with IS according to Algorithm \ref{alg:IS_trigger_gamma} for Gamma distributed losses or Algorithm \ref{alg:IS_trigger_lognormal} for Lognormally distributed losses. \\
    }
    \Else{
        Compute trigger probability $\theta_i$ using standard Monte Carlo simulation according to Algorithm \ref{alg:MC_trigger}. \\
    }
    Compute present value of non-defaultable coupon payments using interest rate model according to Eq.\eqref{zero coupon bond equation}: $P_{\rm Z}(C_i, 0, t_i)$ \\
    Add the present value of coupon payments to the CAT bond price:
    $C_{\rm B} \leftarrow C_{\rm B} + P_{\rm Z}(C_i, 0, t_i) \times \left(1 - \theta_i\right)$. \\
}

\textbf{Final Payment Adjustment:} \\
\If{$\mathbb{E}_{\mathbb{Q}}[L(T)] < D$}{
        Compute trigger probability $\theta_T$ at at maturity $T$ using Monte Carlo simulation with IS. \\
    }
    \Else{
        Compute trigger probability $\theta_T$ using standard Monte Carlo simulation.\\
    }
Compute non-defaultable principal repayments value $P_{\rm Z}(F, 0, T)$, then:
$C_{\rm B} \leftarrow C_{\rm B} + P_{\rm Z}(F, 0, T) \times \left(1 - \theta_T\right)$. \\

\Return CAT Bond Price $C_{\rm B}$.
    
\end{algorithm}

\begin{algorithm}[H]
\caption{Training a Neural Network for CAT Bond Pricing}
\label{alg:cat_bond_nn}
\KwIn{
\\
Price data: Simulated CAT bond prices (computed according to Algorithm \ref{alg:CAT_price}), interest rates, catastrophe event probability parameters, maturities, coupon payment parameters; \\
Model hyperparameters: Number of hidden layers and neurons, activation function, batch normalization flag, dropout rate, learning rate; \\
Training parameters: Number of iterations $N_{\text{Iter}}$, batch size, optimization algorithm (Adam); \\
Regularization parameters: L2 regularization parameter, dropout fraction.
}

\textbf{Initialize:} \\
Construct the neural network model $\mathcal{NN}_{\text{CAT}}$ with the given architecture and hyperparameters. \\
Randomly initialize weights and biases. \\

\For{$i = 1 $ \KwTo $ N_{\text{Iter}}$}{
    \textbf{Sample Training Data:} \\
    Extract a batch of CAT bond features $X_{\text{batch}}=({r_0}_{\text{batch}},\lambda_{\text{batch}},T_{\text{batch}},F_{\text{batch}},D_{\text{batch}})$ and corresponding prices $P_{\text{batch}}:=(P_1, \cdots, P_{N_{\text{batch}}})$.

    \textbf{Forward Propagation:} \\
    Compute predicted CAT bond prices:
    \[
    \hat{P}_{\text{batch}} = \mathcal{NN}_{\text{CAT}}(X_{\text{batch}}):=(\hat{P}_1, \cdots, \hat{P}_{N_{\text{batch}}})
    \]

    \textbf{Compute Loss:} \\
    Calculate the Mean Squared Error (MSE) loss:
    \[
    \mathcal{L} = \frac{1}{N_{\text{batch}}} \sum_{j=1}^{N_{\text{batch}}} (P_j - \hat{P}_j)^2
    \]

    \textbf{Backpropagation and Weight Update:} \\
    Compute gradients of $\mathcal{L}$ with respect to model parameters using automatic differentiation. \\
    Update weights using Adam optimizer (\cite{kingma2014adam}).

    \textbf{Regularization and Constraints:} \\
    Apply L2 weight decay to prevent overfitting. \\
    Use dropout if enabled.
}

\textbf{End}\\

\KwOut{Trained neural network $\mathcal{NN}_{\text{CAT}}$ optimized for CAT bond price prediction.}
\end{algorithm}

\section{Proofs}

\subsection{Proof of Lemma \ref{ATSM_continuous}}
\begin{proof}[Proof of Lemma \ref{ATSM_continuous}]
We begin by considering the first line of \eqref{ATSM_functions}, i.e., the Riccati-type ordinary differential equation (ODE) describing the dynamics of $B$:
\begin{align*}
\frac{ \dee}{ \dee t}B(t,T) = -1 - \alpha(t) B(t,T) + \frac{1}{2} \gamma(t) B^2(t,T), \qquad B(T,T) = 0.
\end{align*}
Since $\alpha(\cdot)$ and $\gamma(\cdot)$ are continuous by Assumption \ref{asu:ATSM_continuous}, standard results on Riccati equations guarantee existence of a unique and continuous solution $[0,T] \ni t \mapsto B(t,T)$ for fixed $T$, see, e.g., \cite{ince2012ordinary}.

We now investigate the continuity of $B(t,T)$ with respect to the maturity parameter $T$. Let $\varepsilon>0$, and define for fixed $T>0$ the difference
\begin{equation*}
[0,T] \ni t \mapsto D(t, \varepsilon):=B(t,T) - B(t, T+\varepsilon).    
\end{equation*}
We compute the derivative of $D(t, \varepsilon)$ with respect to $t$ using the ODE satisfied by both $B(t,T)$ and $B(t,T+\varepsilon)$:
\begin{equation}\label{eq:dDt}
\begin{aligned}
\frac{\dee}{\dee t} D(t, \varepsilon) &=  \frac{1}{2} \gamma(t) \Big( B^2(t,T) - B^2(t, T+\varepsilon) \Big) - \alpha(t) \Big(B(t, T) - B(t,T+\varepsilon) \Big) \\
&= \frac{1}{2} \gamma(t) \Big(B(t, T+\varepsilon) + B(t,T) \Big) D(t, \varepsilon) - \alpha(t) D(t, \varepsilon) \\
&= \Big( \frac{1}{2} \gamma(t) B(t, T+\varepsilon) + B(t,T) - \alpha(t) \Big)D(t, \varepsilon). 
\end{aligned}
\end{equation}
Next, we define
\begin{equation*}
[0,T] \ni t \mapsto M(t):= \frac{1}{2} \gamma(t) \Big| B(t, T+\varepsilon) + B(t,T) \Big| + |\alpha(t)|,   
\end{equation*}
and apply Gr{\"o}nwall's inequality to \eqref{eq:dDt}, which implies:
\begin{equation}\label{eq:appendix_proof1}
|D(t, \varepsilon)| \leq |D(T, \varepsilon)| \exp \Big(\int_t^T M(s) \dee s \Big) \text{ for all } 0 <t < T. 
\end{equation}
We then observe that
\begin{equation*}
C_1:=\exp \Big(\int_0^T M(s) \dee s \Big)< \infty, 
\end{equation*}
which is finite because $t\mapsto \alpha(t)$, $t\mapsto \gamma(t)$, $t\mapsto B(t,T)$, and $t\mapsto B(t,T+\varepsilon)$ are all continuous, and hence $t \mapsto M(t)$ is bounded on $[0,T]$.
Note that by definition,
\begin{equation*}
D(T, \varepsilon) = B(T,T) - B(T, T+\varepsilon) = - B(T, T+\varepsilon),   
\end{equation*}
since $B(T,T)=0$ by the boundary condition of the Riccati ODE \eqref{ATSM_functions}. As the solution $t\mapsto B(t,T+\varepsilon)$ is continuous, we have
\begin{equation*}
|D(T, \varepsilon)| \leq \int_T^{T+\varepsilon} \Big| -1 - \alpha(t) B(t,T+\varepsilon) + \frac{1}{2} \gamma(t) B^2(t,T+\varepsilon) \Big| \D t =: C_2(\varepsilon) < \infty,
\end{equation*}
with $\lim_{\varepsilon \downarrow 0}C_2(\varepsilon)=0$.
Thus, \eqref{eq:appendix_proof1} gives
\begin{equation*}
|D(t, \varepsilon)| \leq C_1 C_2(\varepsilon),    
\end{equation*}
which tends to zero as $\varepsilon \rightarrow 0$. Therefore, $B(t,T)$ is continuous in the maturity $T$.
Next, we want to show that $A(t,T)$ is continuous with respect to the maturity $T$. Consider the defining ODE for $A$
\begin{align*}
\frac{ \dee A(t,T)}{\dee t} = \beta(t) B(t,T) - \frac{1}{2} \delta(t) B^2(t,T),
\end{align*}
where $ t \mapsto \beta(t)$ and $ t \mapsto \delta(t)$ are assumed to be continuous functions and $B(t,T)$ is known to be continuous in $T$ from the above result. Thus, a unique continuous solution $[0,T] \ni t \mapsto A(t, T)$ exists for any $T>0$. 
Fix $\varepsilon>0$, and define the difference of $A(t,T)$ and $A(t,T+\varepsilon)$,
\begin{align*}
[0,T] \ni t \mapsto D_A(t, \varepsilon):&=A(t,T) - A(t, T+\varepsilon)\\
&= \int_T^{T+\varepsilon} \beta(t) B(t,T+\varepsilon)-\tfrac{1}{2} \delta(t) B^2(t,T+\varepsilon)\D t\\
&\hspace{1cm}-\int_T^{T+\varepsilon} \beta(t) B(t,T)-\tfrac{1}{2} \delta(t) B^2(t,T)\D t,
\end{align*}
and hence, since $t \mapsto \beta(t), t\mapsto \delta(t), t \mapsto B(t,T+\varepsilon)$ and $t \mapsto B(t,T)$ all are continuous we obtain $\lim_{\varepsilon \downarrow 0} D_A(t, \varepsilon)=0$ for all $t \in [0,T]$, showing the continuity of $A$ in $T$.
\end{proof}
% {Check for continuity of $A(t,T)$ and $B(t,T)$:
% Firstly, we have the equation for $B(t,T)$ is a Riccati equation:
% \begin{align*}
% \frac{ \dee B(t,T)}{ \dee t} = -1 - \alpha(t) B(t,T) + \frac{1}{2} \gamma(t) B^2(t,T).
% \end{align*}
% This is a first-order nonlinear differential equation. If $\alpha(t)$ and $\gamma(t)$ are continuous, then standard results on Riccati equations guarantee a unique and continuous solution. Moreover, the initial condition $B(T,T)=0$ guarantees that $B(t,T)$ approaches 0 as $T \rightarrow t$.
% Then we have $A(t,T)$ is linear in $A(t,T)$:
% \begin{align*}
% \frac{ \dee A(t,T)}{\dee t} = \beta(t) B(t,T) - \frac{1}{2} \delta(t) B^2(t,T).
% \end{align*}
% If $\beta(t)$ and $\delta(t)$ are also continuous, together with continuous $B(t,T)$, then it guarantees that $A(t,T)$ is a continuously differentiable function. Thus the right-hand side of this equation is continuous with respect to $T$. Moreover, the boundary condition $A(T,T)=0$ ensures that $A(t,T)$ starts at zero when $T=t$.
% In conclusion, with Assumption \ref{asu:Lipschitz} comes true, then the differential equations ensure that $A(t,T)$ and $B(t,T)$ are also continuous.
% }

\subsection{Proof of Lemma \ref{lem:double_stochastic_poisson}}
\begin{proof}[Proof of Lemma \ref{lem:double_stochastic_poisson}]
According to the definition of doubly stochastic Poisson processes, see, e.g., \cite{cox1980point}, we need to show that for $h>0$ the following two properties are fulfilled.
\begin{enumerate}
    \item $\mathbb{Q}\bigl(N(t+h) - N(t) \geq 1 | \mathcal{F}(t) \bigr) = \lambda(t)\Bigl(1-F_X\bigr(D - L(t)\bigr)\Bigr)I\bigl(L(t) < D\bigr)h + o(h)$, \label{enum:proof_1}
    \item $\mathbb{Q}\bigl(N(t+h) - N(t) \geq 2 | \mathcal{F}(t) \bigr) = o(h)$,\label{enum:proof_2}
\end{enumerate}
or equivalently
\begin{enumerate}
    \item[(1$^{\prime}$)] $\lim_{h \downarrow 0} \frac{\mathbb{Q}\bigl(N(t+h) - N(t) \geq 1 | \mathcal{F}(t) \bigr)}{h} = \lambda(t)\Bigl(1-F_X\bigr(D - L(t)\bigr)\Bigr)I\bigl(L(t) < D\bigr)$, \label{enum:proof_1_equiv}
    \item[(2$^{\prime}$)] $\lim_{h \downarrow 0} \frac{\mathbb{Q}\bigl(N(t+h) - N(t) \geq 2 | \mathcal{F}(t) \bigr)}{h} = 0$.\label{enum:proof_2_equiv}
\end{enumerate}
We start by showing \eqref{enum:proof_2}.
Let $t>0$, and consider the time interval $(t, t+h]$ for small $h > 0$. Then, we compute
\begin{align}
&\mathbb{Q}\bigl(N(t+h) - N(t) \geq 2 | \mathcal{F}(t) \bigr) \nonumber\\
&= \mathbb{E}_{\mathbb{Q}}\left[ I \left(M(t+h) - M(t) \geq 2\right) \cdot  I\left(\sum_{i=M(t) + 1}^{M(t+h)} X_i \geq D - L(t)\right) \cdot \underbrace{I\bigl(L(t) < D\bigr)}_{\text{$\mathcal{F}(t)$-measurable}} \Bigg| \mathcal{F}(t) \right] \nonumber \\
&= \left(\sum_{k=2}^{\infty}{\mathbb{E}_{\mathbb{Q}}\left[ \underbrace{I \left(M(t+h) - M(t) = k \right) \cdot I\left(\sum_{i=1}^{k} X_i \geq D - L(t)\right)}_{\text{conditionally independent}} \Bigg| \mathcal{F}(t) \right]}\right) \cdot I\bigl(L(t) < D\bigr) \nonumber \\
&= \left(\sum_{k=2}^{\infty}{\mathbb{E}_{\mathbb{Q}}\left[I \left(M(t+h) - M(t) = k \right) | \mathcal{F}(t) \right] \cdot \underbrace{\mathbb{E}_{\mathbb{Q}}\left[ I\left(\sum_{i=1}^{k} X_i \geq D - L(t)\right) \Bigg| \mathcal{F}(t) \right]}_{\leq 1}}\right) \cdot \underbrace{I\bigl(L(t) < D\bigr)}_{\leq 1} \nonumber \\
&\leq \sum_{k=2}^{\infty}{\mathbb{E}_{\mathbb{Q}}\left[I\left(M(t+h) - M(t) = k \right) | \mathcal{F}(t) \right]} \nonumber \\
&= \sum_{k=2}^{\infty}{\mathbb{Q}\left(M(t+h) - M(t) = k | \mathcal{F}(t)\right)} =   \mathbb{Q}\left(M(t+h) - M(t) \geq  2 | \mathcal{F}(t)\right) = o(h).\label{2 jump proof}
\end{align}
Next, we show property \eqref{enum:proof_1}. Using the above Equation \eqref{2 jump proof}, we get 
\begin{align}
&\mathbb{Q}\bigl(N(t+h) - N(t) \geq 1 | \mathcal{F}(t) \bigr) \nonumber \\
&= \mathbb{Q}\bigl(N(t+h) - N(t) = 1 | \mathcal{F}(t) \bigr) + \mathbb{Q}\bigl(N(t+h) - N(t) \geq 2 | \mathcal{F}(t) \bigr) \nonumber \\
&= \mathbb{Q}\bigl(N(t+h) - N(t) = 1 | \mathcal{F}(t) \bigr) + o(h) \nonumber \\
&= \mathbb{E}_{\mathbb{Q}}\left[ I \left(M(t+h) - M(t) = 1 \right) \cdot  I\left(\underbrace{X_1}_{\text{$X_i$'s are i.i.d.}} \geq D - L(t)\right) \cdot \underbrace{I\bigl(L(t) < D\bigr)}_{\text{$\mathcal{F}(t)$-measurable}} \Bigg| \mathcal{F}(t) \right] + o(h) \nonumber \\
&= \underbrace{\mathbb{Q}\left(M(t+h) - M(t) = 1| \mathcal{F}(t)\right)}_{\lambda(t) \cdot h} \cdot \underbrace{\mathbb{Q}\left(X_1 \geq D - L(t)\right)}_{\Bigl(1-F_X\bigr(D - L(t)\bigr)\Bigr)} \cdot I\bigl(L(t) < D\bigr) + o(h) \nonumber \\
&= \left[ \lambda(t) \cdot \Bigl(1-F_X\bigr(D - L(t)\bigr)\Bigr) \cdot I\bigl(L(t) < D\bigr) \right] \cdot h + o(h)
\end{align}
\end{proof}

\subsection{Proof of Corollary \ref{cor:cat_bond_pricing}}
\begin{proof}[Proof of Corollary \ref{cor:cat_bond_pricing}]
\begin{enumerate}
\item[(i)] By risk-neutral valuation\footnote{Compare the exposition in \cite{bielecki2013credit} for the analogue situation of credit sensitive bonds.}, the arbitrage-free price of a zero coupon CAT bond with face value $F$ and maturity $T$ at time $t \leq T$ is  given by 
\begin{align}
C_{\rm Z}(F, t, T)&= \mathbb{E}_{\mathbb{Q}}\left[P_{\rm Z}(I(\tau>T)F, t, T) | \mathcal{F}(t)\right] \nonumber \\
% && \text{(by Equation \eqref{zero coupon defaultable bond equation})}
&= \mathbb{E}_{\mathbb{Q}}\left[P_{\rm Z}\left(\bigl(1-N(T)\bigr)F, t, T\right) | \mathcal{F}(t)\right] \nonumber \\ 
&= \mathbb{E}_{\mathbb{Q}}\left[P_{\rm Z}\left(F, t, T\right) - P_{\rm Z}\left(N(T)F, t, T\right) | \mathcal{F}(t)\right] \nonumber \\ 
&= P_{\rm Z}\left(F, t, T\right) - \mathbb{E}_{\mathbb{Q}}\left[N(T)P_{\rm Z}\left(F, t, T\right) | \mathcal{F}(t)\right] \nonumber \\ 
&=  \mathbb{E}_{\mathbb{Q}}\left[(1-N(T)P_{\rm Z}\left(F, t, T\right)) | \mathcal{F}(t)\right] P_{\rm Z}\left(F, t, T\right) \nonumber \\ 
&= \mathbb{E}_{\mathbb{Q}}\left[1-\int_{t}^{T}{\lambda(s)\Bigl(1-F_X\bigr(D - L(s)\bigr)\Bigr)I\bigl(L(s) < D\bigr) \,\dee s} \Big| \mathcal{F}(t) \right]P_{\rm Z}(F, t, T), \label{CAT_price1}
\end{align} 
where the last equality is due to Lemma \ref{lem:double_stochastic_poisson} and the properties of doubly stochastic Poisson processes.

\item[(ii)] After we get the price of zero-coupon CAT bond from (1), the price of a coupon CAT bond can be formed directly by considering each coupon payment $C_i$ at time $t_i$ for $i = 1, 2, \ldots, n$ and by using linearity of the expectation.
% The price of a coupon CAT bond with face value $F$, maturity $T$, and coupon payments $C_i$ at time $t_i$ for $i = 1, 2, \ldots, n$, where $0 < t_1 < t_2 < \cdots < t_n \leq T$ at time $t \leq T$ is 
% \begin{equation}
% C_{\rm B}(F, \{C_i\}_{i=1}^n, t, T)= C_{\rm Z}(F, t, T)  + \sum_{i=1}^n {I(t \leq t_i)C_{\rm Z}(C_i, t, t_i)}. \label{CAT_price2}
% \end{equation} 

% Moreover, according to Eq.~\eqref{recoverable bond equation}, we have that the price of a CAT bond with face value $F$, maturity $T$, coupon payments $C_i$ at time $t_i$ for $i = 1, 2, \ldots, n$, where $0 < t_1 < t_2 < \cdots < t_n \leq T$, and recovery rate $R$ at time $t \leq T$ is 
\item[(iii)] The price of a CAT bond with face value $F$, maturity $T$, coupon payments $C_i$ at time $t_i$ for $i = 1, 2, \ldots, n$, where $0 < t_1 < t_2 < \cdots < t_n \leq T$, and recovery rate $R$ at time $t \leq T$ is 
\begin{align}
C_{RB}(F, \{C_i\}_{i=1}^n, t, T) &= C_{B}(F(1-\E_{\Q}[R]), \{C_i\}_{i=1}^n, t, T) + P_{\rm Z}(\E_{\Q}[R]F, t, T) \nonumber \\
&= C_{\rm Z}(F(1-\E_{\Q}[R]), t, T) + \sum_{i=1}^n {I(t \leq t_i)C_{\rm Z}(C_i, t, t_i)} + P_{\rm Z}(\E_{\Q}[R]F, t, T) \nonumber \\
&= \sum_{i=1}^n {I(t \leq t_i)C_{\rm Z}(C_i, t, t_i)} + (1-\E_{\Q}[R])C_{\rm Z}(F, t, T) + \E_{\Q}[R]P_{\rm Z}(F, t, T)\label{CAT_price3}
\end{align}

\end{enumerate}
    
\end{proof}

\subsection{Proof of Proposition \ref{prop_variance_reduction}}
\begin{proof}[Proof of Proposition \ref{prop_variance_reduction}]
First note that $\E[\hat{\theta}^{\rm{MC}}]=\E[\hat{\theta}^{\rm{IS}}]$ follows directly from definition. Moreover, recall that the variance of estimated trigger probability under standard Monte Carlo simulation and under IS is given by
\begin{align*}
\text{Var}(\hat{\theta}^{\rm{MC}}) &= \frac{1}{n} \Big( \mathbb{E}_\mathbb{Q} 
\big[ I(L(t) \geq D) \big) - \theta^2  \Big] 
= \frac{1}{n} ( \theta - \theta^2 ), \\
\text{Var}(\hat{\theta}^{\rm{IS}}) &=\frac{1}{n} \Big(\mathbb{E}_\mathbb{Q} \big[ R(L(t)) I(L(t) \geq D) \big] - \theta^2  \Big).
\end{align*}
Thus, comparing $\text{Var}(\hat{\theta}^{\rm{MC}})$ with $\text{Var}(\hat{\theta}^{\rm{IS}})$ reduces to comparing $\mathbb{E}_\mathbb{Q} \big[ I(L(t) \geq D) \big]$ with $\mathbb{E}_\mathbb{Q} \big[ R(L(t)) I(L(t) \geq D) \big]$.
%%%%%%%%%%%%%%%%%%%%%  Gamma  %%%%%%%%%%%%%%%%%%%%%%%%%%
%Because we consider IS twice for the Poisson process and the loss size separately, given they are independent, we first have the likelihood ratio for the Poisson process is
Note that the likelihood ratio for the Poisson process $\{ M(t)\}_{t \geq 0}$ when applying exponential tilting is
\begin{equation*}
R_{\text{poi}}(L(t)) = \exp\left(\lambda t (e^a - 1) - a M(t)\right).    
\end{equation*}
First, consider $X_i \overset{\text{i.i.d.}}{\sim} \text{Gamma}(1, \beta)$. The likelihood ratio of $L(t)$ computes as
\begin{equation*}
R_{\text{gam}}(L(t)) = (1-\beta b)^{- M(t)} \exp{ \Big(- b L(t) \Big)}.    
\end{equation*}
The condition $a = \tfrac{1}{2} \ln\left(\tfrac{D}{\lambda t \beta} \right)$ writes equivalently as 
\begin{equation}\label{eq:e_a}
e^a = \sqrt{\frac{D}{\lambda t \beta}}.
\end{equation}
Thus, by \eqref{IS_gamma_loss}, \eqref{eq:e_a}, and through setting $b=\frac{1}{\beta}-\frac{\lambda t e^a}{D}$ which is equivalent to $\E_{\Q_{a,b}}[L] = \E_{\Q_{a,b}}[M(t)] \E_{\Q_{a,b}}[X_1] = D$, we obtain
\begin{equation}\label{eq:1_b}
\quad 1-b\beta  = \frac{\lambda t e^a \beta}{D} = \sqrt{\frac{\lambda t \beta}{D}}.   
\end{equation}
Note that the likelihood ratio for the compound Poisson process $\{L(t)\}_{t \geq 0}$ computes now as
\begin{equation*}
R(L(t))= R_{\text{poi}}(L(t)) R_{\text{gam}}(L(t)) = \exp\left(\lambda t (e^a - 1) - a M(t)\right) (1-\beta b)^{- M(t)} \exp{ \big( -b L(t) \big)}.  
\end{equation*}
Taking the expectation, we obtain the bound: 
\begin{align*}
\mathbb{E}_\mathbb{Q} \big[ R(L(t)) I(L(t) \geq D) \big] &= \mathbb{E}_\mathbb{Q} \big[ R(L(t)) \big| L(t) \geq D \big] \cdot  \mathbb{E}_\mathbb{Q} \big[ I(L(t) \geq D) \big]  \\
& \leq \exp \left( \lambda t (e^a - 1) -b D \right) \mathbb{E}_\mathbb{Q} \Big[ e^{- a M(t)} (1-\beta b)^{-M(t)} \Big| L(t) \geq D \Big] \cdot \mathbb{E}_\mathbb{Q} \big[ I(L(t) \geq D) \big].    
\end{align*}
To simplify the first factor of the above expression we use \eqref{eq:e_a}, \eqref{eq:1_b}, and the inequality of arithmetic and geometric means \footnote{AM - GM inequality: $\frac{x+y}{2} \geq \sqrt{xy}$ for two non-negative numbers $x$ and $y$, with equality if and only if $x=y$.} applied to $\sqrt{\frac{D}{\beta}\cdot \lambda t}$ to obtain
\begin{align*}
\lambda t (e^a - 1) -b D &= \lambda t \Big( \sqrt{\frac{D}{\lambda t  \beta}} - 1 \Big) - \Big( 1- \sqrt{\frac{\lambda t  \beta}{D}} \Big) \frac{D}{\beta} \\
&= \sqrt{\frac{D \lambda t}{\beta}} - \lambda t - \frac{D}{\beta} + \sqrt{\frac{D \lambda t }{\beta}} \leq \frac{D}{\beta} \left(\frac{1}{2} + \frac{1}{2} -1\right) = 0
\end{align*}
Hence, $\exp{ \big( \lambda t (e^a - 1) -b D \big)} \leq 1$.
For the second factor, applying \eqref{eq:e_a} and \eqref{eq:1_b} gives
\begin{equation*}
e^{- a M(t)} (1-\beta b)^{-M(t)} = \Big( \sqrt{\frac{\lambda t  \beta}{D}} \Big)^{M(t)(1-1)} = 1
\end{equation*}
Altogether, we have 
\begin{equation*}
\mathbb{E}_\mathbb{Q} \big[ R(L(t)) I(L(t) \geq D) \big] \leq \mathbb{E}_\mathbb{Q} \big[ I(L(t) \geq D) \big],    
\end{equation*}
which provides $\text{Var}(\hat{\theta}^{\rm{MC}}) \geq \text{Var}(\hat{\theta}^{\rm{IS}})$ if losses are Gamma distributed.

\end{proof}

%%%%%%%%%%%%%%%%%%%%%%%%%  Lognormal  %%%%%%%%%%%%%%%%%%%%%%%%%%%%%%%%%%%
%%%%%%%%%%%%%%%%%%%%%%%%%%%%%%%%%%%%%%%%%%%%%%%%%%%%%%%%%%%%%%%%%%%%%%%%%
%%%%%%%%%%%%%%%%%%%%%%%%%%%%%%%%%%%%%%%%%%%%%%%%%%%%%%%%%%%%%%%%%%%%%%%%%

\subsection{Proof of Proposition \ref{prop_variance_reduction_logn}}
\begin{proof}[Proof of Proposition \ref{prop_variance_reduction_logn}]
First note that $\E[\hat{\theta}^{\rm{MC}}]=\E[\hat{\theta}^{\rm{IS}}]$ follows directly from the definition of the estimators.
Our goal is to minimize $\mathbb{E}_\mathbb{Q} \big[ R(L(t)) I(L(t) \geq D) \big]$ to find the optimal IS parameters $a$ and $b$, where 
\begin{equation}\label{eq:R_decom_proof}
R(L(t)) = \exp \Big(\lambda t (e^a - 1) - a M(t) \Big)\exp \Big(\frac{b^2} {2\sigma^2} M(t)\Big) \exp \Big( - \frac{b} {\sigma^2} \sum_{i=1}^{M(t)} \big(\ln (X_i) -\mu \big)\Big).    
\end{equation}
To get rid of the last factor in \eqref{eq:R_decom_proof}, we consider an additional change of measure from the original measure $\mathbb{Q}$ to measure $\widetilde{\mathbb{Q}}$ so that $X_i \overset{\text{i.i.d.}}{\sim} \text{Logn}(\mu-b, \sigma^2)$ under measure $\widetilde{\mathbb{Q}}$. We compute the likelihood ratio 
\begin{equation}\label{eq:R_decom_2}
\begin{aligned}
R_2(L(t)) = \frac{\D \,\widetilde{\mathbb{Q}}}{\D \,\mathbb{Q}} &= \exp\Bigg( \frac{\sum_{i=1}^{M(t)} \big(\ln (X_i) -\mu +b \big)^2}{2 \sigma^2} - \frac{\sum_{i=1}^{M(t)} \big(\ln (X_i) -\mu \big)^2}{2 \sigma^2} \Bigg) \\
&= \exp\Bigg( \frac{b}{\sigma^2}\sum_{i=1}^{M(t)} \big(\ln (X_i) -\mu \big) + \frac{b^2}{2 \sigma^2} M(t) \Bigg).
\end{aligned}
\end{equation}
Then we can rewrite the conditional expectation $\mathbb{E}_\mathbb{Q} \big[ R(L(t)) I(L(t) \geq D) \big]$ under the new measure $\widetilde{\mathbb{Q}}$ as
\begin{equation}\label{eq:ineq_proof_prop_32_1}
\begin{aligned}
\mathbb{E}_\mathbb{Q} \big[ R(L(t)) I(L(t) \geq D) \big] &= \mathbb{E}_{\widetilde{\mathbb{Q}}} \big[ R(L(t)) \cdot R_2(L(t)) I(L(t) \geq D) \big] \\
&= \mathbb{E}_{\widetilde{\mathbb{Q}}} \Big[ \exp \Big(\lambda t (e^a-1) + \big(\frac{b^2}{\sigma^2} - a \big)M(t) \Big) I(L(t) \geq D) \Big] \\
&= \exp \big(\lambda t (e^a-1) \big)\mathbb{E}_{\widetilde{\mathbb{Q}}} \Big[ \exp \Big(\big(\frac{b^2}{\sigma^2} - a \big)M(t) \Big) I(L(t) \geq D) \Big] \\
&= \exp \Big(\lambda t (e^a + e^{\frac{b^2}{\sigma^2} - a} -2) \Big) \cdot
\widetilde{\mathbb{Q}}(L(t) \geq D).
\end{aligned}
\end{equation}

Instead of minimizing the conditional expectation $\mathbb{E}_\mathbb{Q} \big[ R(L(t)) I(L(t) \geq D) \big]$ directly and comparing with $\mathbb{Q}(L(t) \geq D)$, we minimize its ratio $ \frac{\mathbb{E}_\mathbb{Q} \big[ R(L(t)) I(L(t) \geq D) \big]}{\mathbb{Q}(L(t) \geq D)}=\mathbb{E}_\mathbb{Q} \big[ R(L(t)) ~|~L(t) \geq D \big]$.

Note that by the properties of subexponential distributions (to which the Lognormal distribution belongs), we have  for all $n \in \N$, 
\begin{equation}\label{eq:subexponential}
\frac{\mathbb{P}(X_1 + X_2 + \cdots + X_n \geq x)}{\mathbb{P}(X_1 \geq x)} \rightarrow n, \quad \text{as }\, x \rightarrow \infty,    \text{ where } \mathbb{P} \in \{\Q, \widetilde{\Q}\},
\end{equation}
see, e.g., Chapter VI, Section 3 in \cite{Asmussen2007StochasticSA}.
% Next, let $\varepsilon>0$ and choose $\overline{M} \in \N$ such that $$\widetilde{\mathbb{Q}} \left(M(t) \geq \overline{M}\right) < \frac{\varepsilon}{2}.$$
% Due to \eqref{eq:subexponential} there exists some $D_0>0$ such that for all $D >D_0$ and for all  $m \in \N$ with $m \leq  \overline{M}$ we have
% \[
% \widetilde{\mathbb{Q}} \left(\sum_{i=1}^m X_1 \geq D \right) < \left(\frac{\varepsilon}{2}+m \right) \cdot \widetilde{\mathbb{Q}}(X_1 \geq D)=\left(\frac{\varepsilon}{2}+m \right) \cdot \left(1-\Phi \left(\frac{\ln D - \mu + b}{\sigma} \right) \right).
% \]
% This implies for all $D >D_0$
% \begin{equation*}
% \begin{aligned}
% \widetilde{\mathbb{Q}}(L(t) \geq D) &< \sum_{m=1}^{\overline{M}} \widetilde{\mathbb{Q}}(M(t)=m) \cdot \left(\frac{\varepsilon}{2}+m \right) \cdot \widetilde{\mathbb{Q}}(X_1 \geq D)+\sum_{m=\overline{M}+1}^{\infty} \widetilde{\mathbb{Q}}(M(t)=m)  \cdot \widetilde{\mathbb{Q}} \left( \sum_{i=1}^m X_i \geq D \right)\\
% &<  \left(\lambda t + \frac{\varepsilon}{2} \right)\left(1-\Phi \left(\frac{\ln D - \mu + b}{\sigma} \right) \right)+\widetilde{\mathbb{Q}} \left(M(t) \geq \overline{M}\right)\\
% & \leq \lambda t \cdot \left(1-\Phi \left(\frac{\ln D - \mu + b}{\sigma} \right) \right)+\varepsilon.
% \end{aligned}
% \end{equation*}
% where $\Phi(\cdot)$ is the cdf of the standard normal distribution. In the above inequalities, we used $M(t) \sim \text{Poisson}(\lambda t)$ under both measures $\mathbb{P}$ and $\widetilde{\mathbb{Q}}$.
Next, let $\varepsilon>0$ and choose $\overline{M} \in \N$ such that $$\widetilde{\mathbb{Q}} \left(M(t) \geq \overline{M}\right) < \varepsilon.$$
Due to \eqref{eq:subexponential} there exists some $D_0>0$ such that for all $D >D_0$ and for all  $m \in \N$ with $m \leq  \overline{M}$ we have
\begin{align*}
\widetilde{\mathbb{Q}} \left(\sum_{i=1}^m X_1 \geq D \right) &< \left(m+ \varepsilon \right) \cdot \widetilde{\mathbb{Q}}(X_1 \geq D)=\left(m+ \varepsilon \right) \cdot \left(1-\Phi \left(\frac{\ln D - \mu + b}{\sigma} \right) \right), \\
\mathbb{Q} \left(\sum_{i=1}^m X_1 \geq D \right) &> \left(m - \varepsilon \right) \cdot \mathbb{Q}(X_1 \geq D)=\left(m - \varepsilon \right) \cdot \left(1-\Phi \left(\frac{\ln D - \mu}{\sigma} \right) \right).
\end{align*}
This implies for all $D >D_0$
\begin{equation*}
\begin{aligned}
\frac{\widetilde{\mathbb{Q}}(L(t) \geq D)}{\mathbb{Q} \left(L(t) \geq D \right)} &< \sum_{m=1}^{\overline{M}} \widetilde{\mathbb{Q}}(M(t)=m) \cdot \frac{\left(m + \varepsilon \right) \cdot \widetilde{\mathbb{Q}}(X_1 \geq D)}{\left(m - \varepsilon \right) \cdot \mathbb{Q}(X_1 \geq D)} +\sum_{m=\overline{M}+1}^{\infty} \widetilde{\mathbb{Q}}(M(t)=m)  \cdot \frac{\widetilde{\mathbb{Q}} \left( \sum_{i=1}^m X_i \geq D \right)} {\mathbb{Q} \left( \sum_{i=1}^m X_i \geq D \right)} \\
& <  \left(1 + \frac{2\varepsilon}{1-\varepsilon} \right) \cdot \frac{1-\Phi \left(\frac{\ln D - \mu + b}{\sigma} \right)}{1-\Phi \left(\frac{\ln D - \mu}{\sigma} \right)}+\widetilde{\mathbb{Q}} \left(M(t) \geq \overline{M}\right)\\
& < \frac{1-\Phi \left(\frac{\ln D - \mu + b}{\sigma} \right)}{1-\Phi \left(\frac{\ln D - \mu}{\sigma} \right)} + \frac{3\varepsilon}{1-\varepsilon}.
\end{aligned}
\end{equation*}
where $\Phi(\cdot)$ is the cdf of the standard normal distribution. In the above inequalities, we used $\frac{m+\varepsilon}{m-\varepsilon}<1+\frac{2\varepsilon}{1-\varepsilon}$  for all $m \geq 1$, that $\frac{1-\Phi \left(\frac{\ln D - \mu + b}{\sigma} \right)}{1-\Phi \left(\frac{\ln D - \mu}{\sigma} \right)}\leq 1$ for $b \geq 0$, and that $\varepsilon < \frac{\varepsilon}{1-\varepsilon}$. Now, by using \eqref{eq:ineq_proof_prop_32_1}, we derive an upper bound for  our minimization problem by
\begin{align} \label{lognormal_minimize}
\mathbb{E}_\mathbb{Q} \big[ R(L(t)) ~|~L(t) \geq D \big]
= & \exp \Big(\lambda t (e^a + e^{\frac{b^2}{\sigma^2} - a} -2) \Big) \cdot \frac{\widetilde{\mathbb{Q}}(L(t) \geq D)}{\mathbb{Q}(L(t) \geq D)} \nonumber \\
< & \exp \Big(\lambda t (e^a + e^{\frac{b^2}{\sigma^2} - a} -2) \Big) \cdot \left[ \frac{1-\Phi \left(\frac{\ln D - \mu + b}{\sigma} \right)}{1-\Phi \left(\frac{\ln D - \mu}{\sigma} \right)} + \frac{3\varepsilon}{1-\varepsilon} \right] \nonumber \\
< & \exp \Big(\lambda t (e^a + e^{\frac{b^2}{\sigma^2} - a} -2) \Big) \cdot \frac{1-\Phi \left(\frac{\ln D - \mu + b}{\sigma} \right)}{1-\Phi \left(\frac{\ln D - \mu}{\sigma} \right)} + \frac{3\varepsilon}{1-\varepsilon} \cdot C \nonumber \\
=: & f(a,b)+ \frac{3\varepsilon}{1-\varepsilon} \cdot C ~. 
\end{align}
% Now we approximate our minimization problem as
% \begin{align} \label{lognormal_minimize}
% &\mathbb{E}_\mathbb{Q} \big[ R(L(t)) I(L(t) \geq D) \big] \nonumber \\
% &= \exp \Big(\lambda t (e^a + e^{\frac{b^2}{\sigma^2} - a} -2) \Big) \cdot \widetilde{\mathbb{Q}}(L(t) \geq D) \nonumber \\
% < & \exp \Big(\lambda t (e^a + e^{\frac{b^2}{\sigma^2} - a} -2) \Big) \cdot \left[ \lambda t \left(1-\Phi \left(\frac{\ln D - \mu + b}{\sigma} \right) \right)+\varepsilon\right] \nonumber \\
% < & \exp \Big(\lambda t (e^a + e^{\frac{b^2}{\sigma^2} - a} -2) \Big) \cdot  \lambda t \left(1-\Phi \left(\frac{\ln D - \mu + b}{\sigma} \right) \right)+\varepsilon\cdot C:=f(a,b)+\varepsilon\cdot C. 
% \end{align}
for some constant $C$ when restricting to $\underline{a} < a< \overline{a},$ and $ b < \overline{b}$ where $ \overline{a}$ and $ \overline{b}$ are chosen sufficiently large. To minimize $\mathbb{E}_\mathbb{Q} \big[ R(L(t)) ~|~L(t) \geq D \big]$, we will minimize its upper bound which reduces to minimizing $f(a,b)$. To simplify the analysis, we define
\begin{equation} \label{lognormal_ln_minimize}
g(a,b) := \ln f(a,b) = \lambda t (e^a + e^{\frac{b^2}{\sigma^2} - a} -2) + \ln \Big(1-\Phi \big(\frac{\ln D - \mu + b}{\sigma} \big) \Big) - \ln \Big(1-\Phi \big(\frac{\ln D - \mu}{\sigma} \big) \Big).     
\end{equation}
Start by fixing $b \geq 0$, and minimize the function $g(a,b)$ with respect to $a$. The first order condition writes as 
\begin{equation*}
\frac{\partial g}{\partial a} = \lambda t (e^a - e^{\frac{b^2}{\sigma^2} - a}) = 0 \quad \Longrightarrow \quad a^*=\frac{b^2}{2\sigma^2}. 
\end{equation*}
Next, we verify
\begin{equation*}
\frac{\partial^2 g}{\partial a^2}|_{a=a^*} = \lambda t (e^{a^*} + e^{\frac{b^2}{\sigma^2} - a^*}) > 0,
\end{equation*}
which confirms that $a^*$ is indeed the global minimizer of $g(a,b)$ for any fixed $b \geq 0$. Next, we substitute $a^*$ back into function $g(a,b)$ and find the optimal $b^*$. To this end, we define
\begin{equation*}
h(b) := g(a^*,b) = 2 \lambda t (e^{\frac{b^2}{2\sigma^2}} -1) + \ln \Big(1-\Phi \big(\frac{\ln D - \mu + b}{\sigma} \big) \Big) - \ln \Big(1-\Phi \big(\frac{\ln D - \mu}{\sigma} \big) \Big),     
\end{equation*}
and we minimize the function $h(b)$. The first order condition becomes
\begin{equation*}
\frac{\partial h}{\partial b} = \frac{2 \lambda t \, b}{\sigma^2} \, e^{\frac{b^2}{2\sigma^2}} - \frac{1}{\sigma} \, \frac{\phi\Big(\frac{\ln D - \mu + b}{\sigma}\Big)}{1 - \Phi\Big(\frac{\ln D - \mu + b}{\sigma}\Big)} = 0, 
\end{equation*}
where $\phi(\cdot)$ and $\Phi(\cdot)$ are the pdf and cdf of the standard normal distribution. To verify the second order condition, we compute
\begin{equation*}
\frac{\partial^2 h}{\partial b^2} = \frac{2 \lambda t}{\sigma^2} \, e^{\frac{b^2}{2\sigma^2}} \big(1+ \frac{b^2}{\sigma^2} \big) - \frac{1}{\sigma^2} \, \frac{\phi(z)^2-z\phi(z)(1-\Phi(z))}{\big(1 - \Phi(z)\big)^2},
\end{equation*}
where $z=\frac{\ln D -\mu +b}{\sigma}$. This second derivative of function $h$ involves the \emph{Mills ratio}:
\begin{equation*}
R_M(z) = \frac{\phi(z)}{1-\Phi(z)},    
\end{equation*}
so we can rewrite
\begin{equation*}
\frac{\phi(z)^2-z\phi(z)(1-\Phi(z))}{\big(1 - \Phi(z)\big)^2} = \frac{\phi(z)^2}{\big(1 - \Phi(z)\big)^2} - \frac{z\phi(z)}{1 - \Phi(z)} = R_M(z)^2 - z R_M(z) < 1,    
\end{equation*}
according to standard properties of the Mills ratio, see, e.g., \cite[Chapter 13]{johnson1995continuous}. Thus, we have 
\begin{equation*}
\frac{\partial^2 h}{\partial b^2} > \frac{2 \lambda t}{\sigma^2} \, e^{\frac{b^2}{2\sigma^2}} \big(1+ \frac{b^2}{\sigma^2} \big) - \frac{1}{\sigma^2} >0,
\end{equation*}
which confirms that $b^*$ is the global minimizer of $h(b)$.

% Finally, substituting the optimal values $a^*$ and $b^*$ into the function $f(a,b)$, we obtain
% \begin{equation*}
% \mathbb{E}_\mathbb{Q} \big[ R(L(t)) I(L(t) \geq D) \big] < f(a^*, b^*) +\varepsilon \cdot C < {f(a=0,b=0) +\varepsilon \cdot C= \mathbb{Q}(L(t) \geq D) +\varepsilon \cdot C}. 
% \end{equation*}
% We conclude by letting $\varepsilon \downarrow 0$.

Finally, substituting the optimal values $a^*$ and $b^*$ into the function $f(a,b)$, we obtain
\begin{equation*}
\mathbb{E}_\mathbb{Q} \big[ R(L(t)) ~|~(L(t) \geq D) \big]< f(a^*, b^*) +\frac{3\varepsilon}{1-\varepsilon} \cdot C < f(0,0) +\frac{3\varepsilon}{1-\varepsilon} \cdot C = 1 +\frac{3\varepsilon}{1-\varepsilon} \cdot C. 
\end{equation*}
We conclude by letting $\varepsilon \downarrow 0$.

\end{proof}

\subsection{Proof of Proposition \ref{UA_proof}}
\begin{proof}[Proof of Proposition \ref{UA_proof}]
    We need to show that the map 
    $$
    \R^5 \ni (r_0,\lambda,T,F,D) \mapsto C_{\rm Z}(F,T;r_0, \lambda, D)
    $$
    is continuous. Then, the assertion follows with Proposition~\ref{lem_universal}. Firstly, we see that $C_{\rm Z}(F,T;r_0, \lambda, D)$ depends on the face value $F$ and initial interest rate $r_0$ only through $P_{\rm Z}(F, 0, T)$, which by Eq.~\eqref{zero coupon bond equation} is given as  
\begin{equation*}
P_{\rm Z}(F, 0, T)  = F \mathbb{E}_{\mathbb{Q}}\left[e^{-\int_{0}^{T} r_s \,\dee s}\right] =Fe^{A(0,T)-B(0,T)r_0},
\end{equation*}
and which is continuous in $F$ and $r_0$, so thus $C_{\rm Z}(F,T;r_0, \lambda, D)$. Moreover,  $P_{\rm Z}(F, 0, T)$ is continuous in $T$ according to Lemma~\ref{ATSM_continuous}. Define the integral term:
\begin{align*}
g(\lambda, D, T)&:=\mathbb{E}_{\mathbb{Q}}\left[1-\int_{0}^{T}{\lambda\Bigl(1-F_X\bigr(D - L(s)\bigr)\Bigr)I\bigl(L(s) < D\bigr) \,\dee s} \right] \\
&=1 - \int_{\Omega} \int_{0}^{T}{\lambda\Bigl(1-F_X\bigr(D - L(s)(\omega)\bigr)\Bigr)I\bigl(L(s)(\omega) < D\bigr) \,\dee s} \,\dee \mathbb{Q}(\omega)\\
&=1 - \lambda \cdot \int_{0}^{T}\int_{\Omega} \Bigl(1-F_X\bigr(D - L(s)(\omega)\bigr)\Bigr)I\bigl(L(s)(\omega) < D\bigr)  \D \mathbb{Q}(\omega) \D s
\end{align*}
where the last equality follows by Fubini's Theorem. As by Corollary \ref{cor:cat_bond_pricing}, the CAT bond price can be written as 
\[
C_{\rm Z}(F,T;r_0, \lambda, D) = g(\lambda,D,T)P_{\rm Z}(F, 0, T).    
\]
we obtain that $C_{\rm Z}(F,T;r_0, \lambda, D)$ is continuous in $T$ and in $\lambda$. Finally, for the continuity in $D$, let $\varepsilon>0$, and we aim at proving the existence of $\delta>0$ such that 
\begin{equation} \label{continuity in D}
|D_1 - D_2| <\delta \Rightarrow |g(\lambda, D_1, T)-g(\lambda, D_2, T)|<\varepsilon.   
\end{equation}
Without loss of generality, we assume $D_2 > D_1$, and we define
\begin{align*}
A_1&:= \Big\{ (s,\omega) \in \R_+\times \Omega \big| I\bigl(L(s)(\omega) < D_2\bigr) - I\bigl(L(s)(\omega) < D_1\bigr) =1 \Big\} \\
&= \Big\{ (s,\omega) \in \R_+\times \Omega \big| D_1 \leq L(s)(\omega) < D_2 \Big\} \\
A_0&:= \Big\{ (s,\omega) \in \R_+\times \Omega \big| I\bigl(L(s)(\omega) < D_2\bigr) - I\bigl(L(s)(\omega) < D_1\bigr) =0 \Big\} \\
&= \Big\{ (s,\omega) \in \R_+\times \Omega \big| \{L(s)(\omega) < D_1\} \cup \{L(s)(\omega) \geq D_2\} \Big\} .
\end{align*}
Then we can compute the difference
\begin{align} \label{eq_diff_g}
&|g(\lambda, D_1,T)-g(\lambda, D_2,T)|  \nonumber \\
\leq & \mathbb{E}_{\mathbb{Q}}\left[ \int_{0}^{T}{\lambda \Bigl| \bigl(1-F_X\bigr(D_1 - L(s)\bigr)\bigr)I\bigl(L(s) < D_1\bigr) - \bigl(1-F_X\bigr(D_2 - L(s)\bigr)\bigr)I\bigl(L(s) < D_2\bigr) \Bigr| \,\dee s} \right] \nonumber \\
= & \lambda \mathbb{E}_{\mathbb{Q}}\left[ \int_{0}^{T} I_{A_0}{ \Bigl| F_X\bigr(D_1 - L(s)\bigr) - F_X\bigr(D_2 - L(s)\bigr)\Bigr| \,\dee s}  + \int_{0}^{T} I_{A_1}{ \Bigl| 1 - F_X\bigr(D_2 - L(s)\bigr)\Bigr| \,\dee s} \right]. 
\end{align}
Note that $A_1 = \big\{(s,\omega) \in \R_+\times \Omega  \mid \sum_{i=1}^{M(s)} X_i(\omega) \in [D_1,D_2]\big\}$, and we see that for any $s>0$
\begin{equation*}
\Q \Big( \sum_{i=1}^{M(s)} X_i \in [D_1,D_2] \Big) = \sum_{n=0}^\infty \frac{(\lambda s)^n}{n!} e^{-\lambda s} 
\Q \Big( \sum_{i=1}^{n} X_i \in [D_1,D_2]  \Big).
\end{equation*}
Next, one can find $\delta_1 >0$ such that $|D_1-D_2| < \delta_1$ implies
\begin{equation*}
\sup_{n \in \N} \Q \Big( \sum_{i=1}^{n} X_i \in [D_1,D_2] \Big) < \frac{\varepsilon}{2 \lambda T}.
\end{equation*}
Thus, we have 
\begin{equation*}
\int_{\Omega} I_{A_1}(s,\omega) \D \Q(\omega) = \Q \Big( \sum_{i=1}^{M(s)} X_i \in [D_1,D_2] \Big) < \sum_{n=1}^{\infty} \frac{(\lambda s)^n}{n!} e^{-\lambda s} \frac{\varepsilon}{2 \lambda T} = \frac{\varepsilon}{2 \lambda T}.    
\end{equation*}

Since $F_X$ is a cumulative density function for some continuous distribution $X$, we have $\bigl|1 - F_X\bigr(D_2 - L(s)\bigr) \bigr| \leq 1$. Then applying Fubini's Theorem to the second term in \eqref{eq_diff_g}, we have
\begin{align*}
\mathbb{E}_{\mathbb{Q}}\left[\int_{0}^{T} I_{A_1}(s,\omega){ \Bigl| 1 - F_X\bigr(D_2 - L(s)\bigr)\Bigr| \,\dee s} \right] 
&=\int_{\Omega} \int_{0}^{T} I_{A_1}(s,\omega){ \Bigl| 1 - F_X\bigr(D_2 - L(s)(\omega)\bigr)\Bigr| \,\dee s} \,\dee \mathbb{Q}(\omega) \\
&= \int_{0}^{T} \int_{\Omega} I_{A_1}(s,\omega){ \Bigl| 1 - F_X\bigr(D_2 - L(s)(\omega)\bigr)\Bigr| \,\dee \mathbb{Q}(\omega) \,\dee s} \\
&\leq \int_{0}^{T} \int_{\Omega} I_{A_1}(s,\omega){\,\dee \mathbb{Q}(\omega) \,\dee s} < \int_{0}^{T} \frac{\varepsilon}{2 \lambda T} {\,\dee s} = \frac{\varepsilon}{2 \lambda}.
\end{align*}
Moreover, as $d \mapsto F_X(d-L(s)(\omega))$ is continuous for all $\omega	\in \Omega$, there exists $\delta_2 >0$ such that $|D_1-D_2|< \delta_2$ yields
\begin{equation*}
\Bigl| F_X\bigr(D_1 - L(s)\bigr) - F_X\bigr(D_2 - L(s)\bigr)\Bigr| \leq \frac{\varepsilon}{2 \lambda T}   
\end{equation*}
 Thus $|D_1 - D_2| < \delta:= \min \{\delta_1,\delta_2\}$ implies
\begin{equation*}
|g(\lambda, D_1,T)-g(\lambda, D_2,T)|
< \lambda \left(\mathbb{E}_{\mathbb{Q}}\left[ \int_{0}^{T} { \frac{\varepsilon}{2 \lambda T}   \,\dee s} \right] + \int_{0}^{T} \frac{\varepsilon}{2 \lambda T}   {\,\dee s} \right) =\left(\frac{\varepsilon}{2 \lambda }   + \frac{\varepsilon}{2 \lambda }  \right) \lambda  = \varepsilon. 
\end{equation*}
\end{proof}

\begin{proof}[Proof of Proposition \ref{UA_proof_coupon}]
Firstly, we have $C_{\rm Z}(F, T;r_0,\lambda, D)$ is continuous in $(r_0, \lambda, T, F, D)$, according to Proposition~\ref{UA_proof}. Similarly, for each coupon payment dates $t_i$, for $i=1, \cdots, n$, the discounted coupon payments $C_{\rm Z}(C_i, 0, t_i)$ is continuous in $(r_0, \lambda, t_i, C_i, D)$ and does not depend on any other inputs. Recall the formula of coupon CAT bond price that
\begin{equation*}
C_{\rm B}(F, \{C_i\}_{i=1}^n, 0, T)= C_{\rm Z}(F, 0, T)  + \sum_{i=1}^n {C_{\rm Z}(C_i, 0, t_i)},
\end{equation*} 
it is a cumulative sum of $C_{\rm Z}(F, 0, T)$ and $(\lambda, t_i, C_i, D)$. We can conclude that $C_{\rm B}(F, \{C_i\}_{i=1}^n, 0, T)$ is continuous in $(r_0, \lambda,\{t_i\}_{i=1}^n,T,\{C_i\}_{i=1}^n,F,D)$ and the universal approximation theorem from Proposition~\ref{lem_universal} can be applied.
\end{proof}

\end{document}